\documentclass[12pt, draftclsnofoot, onecolumn]{IEEEtran}
\usepackage{epsfig,latexsym}
\usepackage{float}
\usepackage{indentfirst}
\usepackage{amsmath}
\usepackage{amssymb}
\usepackage{times}
\usepackage{subfigure}
\usepackage{psfrag}
\usepackage{hyperref}
\usepackage{cite}
\usepackage{lastpage}
\usepackage{fancyhdr}
\usepackage{color}
 \usepackage{amsthm}
\usepackage{bigints}
\newtheorem{Proposition}{Proposition}
\newtheorem{Lemma}{Lemma}
\newtheorem{Corollary}{Corollary}
\newtheorem{lemma}[Lemma]{$\mathbf{Lemma}$}
\newtheorem{proposition}[Proposition]{Proposition}
\newtheorem{corollary}[Corollary]{$\mathbf{Corollary}$}

\begin{document}
\title{ Impact of Non-orthogonal Multiple Access on the Offloading of Mobile Edge Computing}

\author{ Zhiguo Ding, \IEEEmembership{Senior Member, IEEE}, Pingzhi Fan,   \IEEEmembership{Fellow, IEEE},   and H. Vincent Poor, \IEEEmembership{Fellow, IEEE}  \thanks{  

    Z. Ding and H. V. Poor are  with the Department of
Electrical Engineering, Princeton University, Princeton, NJ 08544,
USA. Z. Ding
 is also  with the School of
Electrical and Electronic Engineering, the University of Manchester, UK (email: \href{mailto:zhiguo.ding@manchester.ac.uk}{zhiguo.ding@manchester.ac.uk}, \href{mailto:poor@princeton.edu}{poor@princeton.edu}).
 
P. Fan is with the Institute of Mobile
Communications, Southwest Jiaotong University, Chengdu, China (email: \href{mailto:pzfan@swjtu.edu.cn}{pzfan@swjtu.edu.cn}).

  }\vspace{-2em}} \maketitle

\begin{abstract} 
This paper considers the coexistence of two important communication techniques, non-orthogonal multiple access (NOMA) and mobile edge computing (MEC). Both NOMA uplink and downlink transmissions are applied to MEC, and analytical results are developed to demonstrate that the use of NOMA can efficiently reduce  the latency and energy consumption of MEC offloading. In addition, various asymptotic studies are   carried out to reveal the impact of the users' channel conditions and transmit powers on the   application  of NOMA to MEC is quite different to those in conventional NOMA scenarios. Computer simulation results are also provided to facilitate the performance evaluation of NOMA-MEC and also verify the accuracy of the developed analytical results. \vspace{-1em}
\end{abstract} 

\section{Introduction}
Non-orthogonal multiple access (NOMA) has been viewed as a key enabling technology in next-generation wireless networks due to its superior spectral efficiency \cite{jsacnomaxmine}. On the one hand,  the principle of NOMA brings fundamental changes to the design of future multiple access techniques \cite{6666156,lpma}. In particular, compared to conventional orthogonal multiple access (OMA) which allocates orthogonal bandwidth resource blocks to users,  NOMA  encourages  the users   to share the same spectrum, where multiple access interference is handled by applying advanced transceiver designs, such as superposition coding and successive interference cancellation (SIC). Hence compared to OMA, NOMA offers better flexibility     for efficiently utilizing the scarce bandwidth resources. 

On the other hand, the principle of NOMA has also been shown  important to the evolution  of many other types of communication techniques. For example, the spectral efficiency of multiple-input multiple-output (MIMO) systems can be significantly improved by designing sophisticated MIMO-NOMA transmission schemes and       harvesting  the spatial degrees of freedom in a more efficient way compared to MIMO-OMA \cite{Zhiguo_mimoconoma, 7582424,7095538}. Another example is the application of   NOMA to millimeter-wave (mmWave) communication systems, and the existing studies show that the directional transmission feature of mmWave propagation is ideal for the application of NOMA, where users with strongly correlated channels are grouped together for the implementation of NOMA \cite{Zhiguo_mmwave}. Wireless caching is one of the latest  examples for the applications of NOMA to other communication techniques, where NOMA assisted content pushing and delivery schemes have  been developed  to improve the cache hit probability and ensure that the files stored in the local caches are frequently updated during on-peak hours \cite{Dingicc18}. 

This paper is to focus on the coexistence of   NOMA and mobile-edge computing (MEC) which is another important communication technique in future wireless networks \cite{6069707,6553297}. The use of MEC is motivated by the fact that emerging mobile applications, such virtual reality, augmented reality, and interactive gaming, make mobile networks computationally constrained. Take virtual reality as an example. Mobile nodes need to carry out object recognition, pose estimation, vision-based tracking, etc. Furthermore, for virtual reality assisted  gaming, the mobile devices are also expected to facilitate mixed reality and human computer interaction \cite{7946930}. However, most mobile devices are computation and power limited, which means that relying on the mobile devices to locally complete the computationally intensive tasks will result in two disadvantages. One is that the batteries of the devices can be drained quickly, and the other is that the devices might not be able to complete the tasks  before their deadlines.  The key idea of MEC is to employ more resourceful computing facilities at the edge of mobile networks, such as access points  and small-cell  base stations integrated with MEC servers, and ask the mobile users to offload their computationally intensive tasks to the MEC facilities. In order to improve the energy-latency tradeoff of MEC, a dynamic computation offloading scheme was proposed  by assuming that  mobile devices can carry out energy harvesting by using renewable energy sources \cite{7572018}. In \cite{8234686} and \cite{8264794} a similar MEC scenario was considered, where   simultaneous wireless information and power transfer was applied to MEC to facilitate the user cooperation.  In \cite{8279411}, a user scheduling scheme was proposed to MEC in order to achieve a balanced tradeoff between the latency and reliability for task offloading.  In \cite{MEC01},  a more challenging multi-user MEC scenario was considered, where the users offload their tasks to the MEC server in an asynchronous manner.  

Initial studies in \cite{8269088} and \cite{8267072} have already demonstrated the benefit for the application of NOMA to   MEC, by developing various optimization frameworks. However, there is still lack of theoretic performance analysis for a better understanding of the impact of NOMA on  MEC, which is the motivation of this paper. The contributions of this paper are listed as follows:
\begin{itemize}
\item The application of NOMA uplink transmission to MEC is considered, where the impact of NOMA on the latency of MEC is foused first. When there are multiple users and a single MEC server, the use of NOMA can ensure that multiple    users   complete their offloading at the same time, which   effectively reduces the offloading latency. The probability for a strong user to complete its offloading by   using the time which would be solely occupied by a weak user in the OMA mode is  characterized first and then used to identify the impact of the users' channel conditions and transmit powers on the offloading latency.  The carried out asymptotic studies reveal that,  in the low signal-to-noise ratio (SNR) regime, it is almost sure that the use of NOMA can guarantee a superior latency performance, i.e.,   the strong user does not need extra time, but just uses the time allocated to the weak user for offloading. However, this conclusion is not valid in the high SNR regime, as shown by the carried out asymptotic studies.  These observations are quite different from conventional NOMA scenarios, where the benefit  of using NOMA is more obvious in the high SNR regime. 

\item From the energy perspective, NOMA-MEC is not energy efficient, if the strong user is forced to complete its offloading by only using     the time which would be solely occupied    by  the weak user in the OMA mode. A more energy efficient offloading approach is to ask the strong user to first offload  parts of its task while the weak user is offloading, and then offload its remaining data to the server by using a dedicated time slot. Our developed analytical results show that the modified NOMA assisted MEC protocol can offload more data than OMA, while using less energy. This conclusion is surprising since it is commonly believed that more energy is needed for many NOMA transmission schemes compared to their OMA counterparts,   in order to combat strong multiple access interference.  

\item The application of NOMA downlink transmission to MEC is also considered, where  a user  uses NOMA to offload its multiple   tasks to multiple MEC servers simultaneously.  For NOMA uplink transmission, admitting an additional node into the system will not bring any performance degradation to the existing nodes, as long as the newcomer's   signal is decoded correctly  at the first step of SIC. However, this is not valid to  NOMA downlink transmission,    which motivates the use of the cognitive radio inspired power allocation policy. The analytical results are developed to demonstrate that NOMA-MEC with cognitive radio power allocation  can simultaneously  reduce the energy consumption for offloading and also increase the amount of data offloaded to the servers, particularly in the high SNR regime. In addition, the carried out asymptotic studies  show that, for the application of NOMA downlink transmission  to MEC, it is important to group  servers   with strong channel conditions, in order to realize  the performance gain of NOMA-MEC over OMA-MEC, whereas, for the application of NOMA uplink transmission to MEC,  it is preferable  to schedule users with diverse channel conditions, i.e., a user with poor channel conditions  is paired with a   user with strong channel conditions.  
\end{itemize}

\section{System Model} 
Consider a general MEC communication scenario with $M$ users and $K$ access points with   integrated MEC servers. All the nodes are assumed to have a single antenna, and operate in the half duplex mode. Each user needs to complete  computationally intensive latency-critical tasks.  Because of the users' limited computation capabilities,   carrying out those   tasks locally can consume a significant amount of time and energy, which is the motivation for the use of MEC. In order to clearly illustrate the impact of NOMA on MEC, the following assumption is made in this paper:

{\bf Assumption 1:} {\it The users always prefer  to offload  their tasks to the MEC servers.}

With this assumption, the cost of using  OMA-MEC  for offloading will be compared to that of NOMA-MEC in this paper,  so   the performance gain of NOMA over OMA can be clearly demonstrated.   Assume that each user has $L$ tasks, where each task is inseparable and  task $l$ belonging to user $m$ containts  $N_{m,l}$ bits. 

Typically, MEC consists of  two following phases. The first phase is the offloading phase, where a user transmits its tasks to one or more than one MEC server. The second phase is the feedback phase, where the MEC servers carry out the offloaded tasks and feed the outcomes of these computations back to the users. In this paper, the impact of NOMA on the first phase of MEC is focused, and the following assumption is used:

{\bf Assumption 2:} {\it The costs for the second phase of MEC are omitted in the paper.} 

Note that in the literature of MEC, this assumption has been commonly used due to the following two reasons \cite{7572018, 8234686, 8264794}. Firstly,  the delay caused by  the second phase of MEC, i.e., the time for a server to compute an offloaded task and the time for a user to download the computation results from a server, is negligible, because of the superior computation capabilities   of the servers as well as  the small sizes of the computation results.  Secondly, the energy for an MEC server to compute the offloaded tasks as well as the  transmission energy consumption during the second phase of MEC can   also be omitted, since the MEC servers are not energy constrained. 

The performance of MEC can be evaluated from the latency  and energy perspectives:
\begin{itemize}

\item {\it Latency of MEC:}
Denote the data rate for user $i$  to offload task $l$ by $R_{i,l}$ which is a function of the used transmit power, denoted by $P^{ow}_{i,l}$. The   time required for offloading task $l $ of user $i$ is given by
\begin{align}
 {T}_{i,l}= \frac{ N_{i,l}}{R_{i,l}} .
\end{align} 
Due to Assumption 1, all the tasks will be offloaded, and hence  there is no delay cost for   local computing. 

\item {\it Energy Consumption of MEC:}
Recall that the offloading transmit power is $P^{ow}_{i,l}$, which is determined by $R_{i,l}$. Therefore, the total energy consumed by offloading all the $L$ tasks of user $i$ is given by
\begin{align}
E_{i} = \sum^{L}_{l=1}P^{ow}_{i,l} \frac{ N_{i,l}}{R_{i,l}},
\end{align}
where the use of Assumption 1 means that there is no energy cost for local computing, and the energy consumption during the second phase of MEC is omitted  due to Assumption~2. 
\end{itemize}

\section{Application of NOMA Uplink Transmission to MEC}\label{section uplink}
This section is to focus on one particular type of MEC scenarios, where  $M$ users offload their tasks to a single MEC server ($K=1$) and each user has a single task for offloading ($L=1$). 
  Offloading in  this MEC scenario can be viewed as a special case of uplink transmission, to which both OMA and NOMA can be applied. Depending on the user's quality of service (QoS) requirements, different   MEC offloading strategies can be applied, as described in the following two subsections.

Without loss of generality, assume that the users are ordered as follows:
\begin{align}
|h_1|^2\leq  \cdots \leq |h_M|^M,
\end{align}
where $h_m$ denotes the channel gain between user $m$ and the MEC server. In this paper, the users' channels are assumed to be quasi-static Rayleigh fading. In order to avoid overloading the MEC server at a single bandwidth resource block, such as a time slot or a frequency channel, we assume that only two users, user $m$ and user $n$,  are scheduled to be served by the MEC server at the same resource block  where $m<n$.

\subsection{Impact of NOMA on Offloading Latency}
If the users' tasks are delay sensitive, i.e., using less offloading time has higher priority than energy consumption,  OMA-MEC and NOMA-MEC can be implemented as follows.  

In OMA-MEC,  the users are allocated with dedicated time slots  for offloading their tasks to the MEC server individually, i.e.,  each user needs the following time interval  for delivering its task  to the server\footnote{For notational simplicity, subscript $l$ is omitted since each user has a single task for offloading ($L=1$).}:
\begin{align}
{T}_i \triangleq   \frac{N}{\log\left(1+\frac{P^{ow}_{i}}{P_N}|h_i|^2\right)},
\end{align}
for $i\in\{m,n\}$, where   $P_N$ denotes   the receive noise.     To facilitate performance analysis, we   assume that the users' tasks have the same size, i.e., $N=N_{i,1}$, for $i\in \{m, n\}$.  

In NOMA-MEC, user $n$ is admitted to time slot $T_m$ which would be solely occupied by  user $m$ in the OMA mode, and user $n$ is asked to  finish its offloading within  $T_m$. Compared to OMA-MEC, the advantage  of  NOMA-MEC   is that user $n$ does not need extra time for offloading, and hence the offloading latency is reduced.   It is important to point out that admitting  user $n$ to time slot $T_m$ does not cause any performance degradation to user $m$, if the user $n$'s signals  are decoded before  user $m$'s at the MEC server and  also user $n$ uses the following rate constraint:~\cite{Cover1991} 
\begin{align}
R_n\leq \log\left(1+\frac{P^{ow}_{n}|h_n|^2}{P^{ow}_{m}|h_m|^2+P_{N}}\right).
\end{align}

The following lemma provides  the closed-form expression for the probability   $\mathrm{P}_n = \mathrm{P}\left( R_nT_m\geq N \right)$, which measures the likelihood of the event that   user $n$ can complete its offloading within $T_m$, for given $P^{ow}_{n}$ and $P^{ow}_{m}$.

\begin{lemma}\label{lemmapn}
 For given $P^{ow}_{n}$ and $P^{ow}_{m}$,    the probability for user $n$ to complete offloading by using the time slot which would be solely occupied by  user $m$ in the OMA mode is given by
\begin{align}\label{eqlemma1}
\mathrm{P}_n   
=& 
c_{mn} \sum^{n-1-m}_{p=0} \frac{c_p}{M-m-p}\sum^{m-1}_{l=0} c_l e^{\frac{b^2}{4a}}  \\\nonumber &\times \frac{\sqrt{\pi}}{2\sqrt{a}}\left(1 - \Phi\left(\frac{\max\{0,\rho_n-\rho_m\}}{\rho^2_m}+\frac{b\sqrt{a}}{2a}\right)\right)+1\\\nonumber &-\frac{M!}{(m-1)!(M-m)!} \sum^{m-1}_{l=0}c_l\frac{e^{-(M-m+l+1)\frac{\max\{0,\rho_n-\rho_m\}}{\rho^2_m} }}{M-m+l+1} ,
\end{align}
where $\rho_i=\frac{P_i^{ow}}{P_N}$, $i\in\{m,n\}$, $c_{mn}=\frac{M!}{(m-1)!(n-1-m)!(M-n)!}$, $c_p={n-1-m \choose p} (-1)^{n-1-m-p}$, $c_l={m-1\choose l} (-1)^l$,  $a=\frac{\rho^2_m}{\rho_n}(M-m-p)$, $b=p+l+1+(M-m-p)\frac{\rho_m}{\rho_n}$, and $\Phi(\cdot)$ denotes the probability integral. 
\end{lemma}
\begin{proof}
Please refer to Appendix A. 
\end{proof}

In order to carry out asymptotic studies, we first present the following proposition which will be used for the development of the high and low SNR approximations for $\mathrm{P}_n $. 
\begin{proposition}\label{lemma2}
For $m<M$, the following equality holds
\begin{align}
\frac{M!}{(m-1)!(M-m)!} \sum^{m-1}_{l=0}c_l \frac{
1 }{ (M-m+l+1)}=1. 
 \end{align}
\end{proposition}
\begin{proof}
Please refer to Appendix B.
\end{proof}

By using Lemma \ref{lemmapn} and Proposition \ref{lemma2}, the high and low SNR approximations for $\mathrm{P}_n   $ can be obtained   in the following lemmas.
\begin{lemma} \label{lemmahigh}
  When both $P_n^{ow}$ and $P_m^{ow}$ approach infinity and $\eta\triangleq \frac{P_n^{ow}}{P_m^{ow}}$ is a constant,   the probability for user $n$ to complete offloading within $T_m$ can be approximated as follows: 
\begin{align}
\mathrm{P}_n   
\approx & 
 \frac{1}{\rho_m^{\frac{m}{2}}}\sum^{n-1-m}_{p=0} \frac{\eta^{\frac{m}{2}}c_{mn}c_p}{(M-m-p)^{\frac{m}{2}+1}}    \left(\tilde{Q}_1- \tilde{Q}_2\right) .
\end{align}
The two parameters, $\tilde{Q}_1$ and $\tilde{Q}_2$, are given by
\begin{eqnarray}
\tilde{Q}_1 \approx  \left\{\begin{array}{ll} \frac{\sqrt{\pi}(-1)^{m-1}(m-1)!,}{\left(\frac{m-1}{2}\right)! 2^{m}} , &\text{if $m$ is an odd number}\\  \frac{\sqrt{\pi}\mu_m}{\left(\frac{m}{2}\right)! 2^{m+1}a^{\frac{1}{2}}}  , & \text{if $m$ is an even number}
  \end{array} \right.,
\end{eqnarray}
and 
\begin{eqnarray}
\tilde{Q}_2 \approx   \left\{\begin{array}{ll}  \frac{ \mu_m}{m!!2^{\frac{m+1}{2}}a^{\frac{1}{2}}}, &\text{if $m$ is an odd number}\\    \frac{ (-1)^{m-1}(m-1)! }{(m-1)!!2^{\frac{m}{2}}} , & \text{if $m$ is an even number}
  \end{array} \right.,
\end{eqnarray}
where $\lambda=\left[p+1+\frac{(M-m-p)}{\eta}\right]$ and $\mu_m=\left(   \sum^{m-1}_{l=0} c_l l^m +  m! \lambda  (-1)^{m-1} 
\right)$. 
\end{lemma}
\begin{proof}
Please refer to Appendix C. 
\end{proof}

\begin{lemma}\label{lemmalow}
At low SNR, i.e., when both $P_n^{ow}$ and $P_m^{ow}$ approach zeor and $\eta$ is a constant,  the probability for user $n$ to complete offloading within $T_m$ approaches one, i.e., $\mathrm{P}_n   \rightarrow1$.
\end{lemma}
\begin{proof} 
Please refer to Appendix C. 
 \end{proof}
 Following steps similar to those in the proof for Lemma \ref{lemmalow}, we can have the following corollary.
 \begin{corollary}\label{corollarylow}
When   $P_n^{ow}$ approaches infinity and $P_m^{ow}$ is fixed,  the probability for user $n$ to complete offloading within $T_m$ approaches one, i.e., $\mathrm{P}_n   \rightarrow1$.
\end{corollary}
{\it Remark 1:} Lemma \ref{lemmahigh} indicates that    $\mathrm{P}_n $ approaches zero at high SNR.  This phenomenon  can be explained   in the following. When $P_m^{ow}$  becomes infinity, user $m$'s rate becomes infinity, and hence $T_m$ approaches zero. On the other hand, the  data rate for user $n$ to transmit during time slot $T_m$ becomes a constant at high SNR, i.e., $\log\left(1+\frac{\rho_n|h_n|^2}{\rho_m|h_m|^2+1}\right)\rightarrow \log\left(1+\frac{\eta |h_n|^2}{ |h_m|^2}\right)$ for $P_n^{ow}\rightarrow \infty$ and $P_m^{ow}\rightarrow \infty$. Therefore, with $T_m\rightarrow 0$ and a constant $R_n$, it will be difficult for     user $n$ to complete offloading within $T_m$. 

{\it Remark 2:} The decay rate of $\mathrm{P}_n $ can be obtained as follows. At high SNR, $a$ also approaches infinity. Therefore, $\tilde{Q}_1$ is dominant when $m$ is an odd number, otherwise $\tilde{Q}_2$ becomes dominant. As a result, $\mathrm{P}_n $ in Lemma \ref{lemmahigh} can be further approximated as follows:
\begin{eqnarray}
   \left\{\begin{array}{ll} 
  \frac{\sum^{n-1-m}_{p=0} \frac{c_{mn}c_p\sqrt{\pi} (m-1)!}{(M-m-p)^{\frac{m}{2}+1}\left(\frac{m-1}{2}\right)! 2^{m}} }{\rho^{\frac{m}{2}}}  , &m\in\{1, 3, \cdots\}\\  \frac{\sum^{n-1-m}_{p=0} \frac{c_{mn}c_p  (m-1)!}{(M-m-p)^{\frac{m}{2}+1}(m-1)!!2^{\frac{m}{2}}}}{\rho^{\frac{m}{2}}}       , & m\in\{2, 4, \cdots\}
  \end{array} \right.,
\end{eqnarray}
which means that the decay rate of  $\mathrm{P}_n $ is $\rho^{-\frac{m}{2}}$, i.e., scheduling a user with poor channel conditions to act as the NOMA weak user is beneficial to increase  $\mathrm{P}_n $.

 {\it Remark 3:} Lemma \ref{lemmalow} indicates that, in the low SNR regime,  it is almost sure that  user $n$ can complete its data offloading by using   $T_m$ only. The reason is that, at low SNR, a user with poor channel conditions needs to use a significant amount of   time for offloading, which provides an ideal opportunity for   using  NOMA, i.e., user $n$ has more time to offload its task to the MEC server. For a similar reason,    another ideal situation for the application of NOMA-MEC is that $P_n^{ow}$ approaches infinity and $P_m^{ow}$ is fixed, as indicated  by   Corollary \ref{corollarylow}.

 {\it Remark 4:} 
If user $n$ completes its offloading within $T_m$, the latency of NOMA-MEC offloading can be significantly reduced,   but at a price that more energy is consumed compared to OMA-MEC.    Particularly, in order to strictly ensure that $N$ bits are offloaded  within $T_m$, the power used by user $n$ needs to satisfy the following constraint:
\begin{align}
\log\left(1+\frac{P^{ow}_n |h_n|^2}{P^{ow}_m |h_m|^2+1}\right)T_m>\log\left(1+P^{ow}_m |h_m|^2\right)T_m. 
\end{align}
Therefore, the minimal power for user $n$ is given by
\begin{align}
P^{ow}_n = \frac{|h_m|^2}{|h_n|^2}P^{ow}_m \left(1+P^{ow}_m|h_m|^2\right).
\end{align}
In OMA, if user $n$ is given the same amount of time ($T_m$) for offloading $N$ bits, user $n$'s power needs to satisfy the following:
\begin{align}
P_n^{OMA} = \frac{|h_m|^2}{|h_n|^2}P^{ow}_m.  
\end{align}
So  the price for the improved latency is for user $n$   to consume more energy, i.e.,  $P^{ow}_n-P_n^{OMA} = \frac{ (P^{ow}_m|h_m|^2 )^2}{|h_n|^2}$.

\subsection{Impact of NOMA on Offloading Energy Consumption} \label{subsection energy}
The energy inefficiency pointed out in  Remark 4 is due to the imposed constraint that user $n$ has to complete its offloading within $T_m$.  By removing this constraint, a modified NOMA-MEC protocol with better energy efficiency can be designed as described in the following.

 In order to have a fair comparison between OMA and NOMA, first consider the following modified OMA benchmark. In particular, assume that each user is allocated an equal-duration time slot with $T$ seconds for  offloading. Furthermore, denote    user $n$'s transmit power   in OMA  by $P^{ow}_i$, $i\in\{m,n\}$, which means that the energy consumption for user $i$ in OMA is $TP^{ow}_i$ and the amount of data sent within $T$ is $ T\log(1+P_i^{ow} |h_i|^2) $.
 
For the modified NOMA-MEC protocol, the two users use NOMA to transmit simultaneously  during the first time slot, and user $n$ solely occupies   the second time slot\footnote{Since   admitting user $n$ into the first time slot will not cause any performance degradation to user $m$,     only user $n$'s performance is focused.}. Assume that   user $n$'s  power  in NOMA is only a portion of that in OMA, i.e.,  $\beta P^{ow}_n$. Therefore,  the overall energy consumption for user $n$ in NOMA is $2T\beta P^{ow}_n$ and   the amount of data sent within $2T$ is given by
\begin{align}
  T\log\left(1+\frac{\beta P_n^{ow} |h_n|^2}{P_m^{ow} |h_m|^2+1}\right) +T\log(1+\beta P_n^{ow} |h_n|^2). 
\end{align}

$\beta$ is an energy reduction parameter and needs to be smaller than $\frac{1}{2}$ since the constraint that 
NOMA-MEC is more energy efficient than OMA-MEC is equivalent to the following: 
\begin{align}\label{power}
2T\beta P^{ow}_n< T  P^{ow}_n.
\end{align} 
However, with    $\beta$ satisfying  \eqref{power},  it is not guaranteed that NOMA-MEC can delivery the same amount of data as OMA-MEC, and the probability for this event can be expressed as follows: 
\begin{align}
\tilde{\mathrm{P}}_n\triangleq &\mathrm{P} \left( T\log\left(1+\frac{ \beta\rho_n  |h_n|^2}{\rho_m|h_m|^2+1}\right)\right.\\\nonumber&+\left.T\log(1+ \beta\rho_n |h_n|^2) \leq T\log(1+\rho_n |h_n|^2) \right).
\end{align} 

The following corollary provides the closed-form expression for $\tilde{\mathrm{P}}_n$. 

\begin{corollary}\label{corollary1}
If  $(1-\beta)\rho_m\geq \beta^2\rho_n$, the probability $\tilde{\mathrm{P}}_n$ can be expressed as follows: 
\begin{align}
\tilde{\mathrm{P}}_n  = &1-c_{mn}\sum^{n-1-m}_{p=0}c_p \sum^{m-1}_{l=0}c_l  \frac{ e^{ 
-(M-m-p)\frac{1-2\beta}{\beta^2\rho_n}
}}{(M-m-p)\tilde{a} },
\end{align}
otherwise
\begin{align}
\tilde{\mathrm{P}}_n  = &1-\frac{M!\sum^{m-1}_{l=0}c_l\frac{e^{-(M-m+l+1)\kappa_1 }}{M-m+l+1}}{(m-1)!(M-m)!} -c_{mn}\sum^{n-1-m}_{p=0}c_p\\\nonumber &\times  \sum^{m-1}_{l=0}c_l   \frac{  e^{ 
-(M-m-p)\frac{1-2\beta}{\beta^2\rho_n}
}\left(1 - e^{- \tilde{a}\kappa_1}\right)  }{\tilde{a}(M-m-p)} ,
\end{align}
where $\kappa_1=\frac{1-2\beta}{\beta^2\rho_n-(1-\beta)\rho_m}$ and
$\tilde{a}=\frac{\rho_m(1-\beta)(M-m-p)}{\beta^2 \rho_n}+p+l+1$.
\end{corollary}
\begin{proof} With some algebraic manipulations,  $\tilde{\mathrm{P}}_n$ can be   rewritten as follows:
\begin{align}
\tilde{\mathrm{P}}_n &=\mathrm{P} \left(  |h_n|^2\leq \frac{(1-\beta)(1+\rho_m|h_m|^2)-\beta}{\beta^2\rho_n}\right).
\end{align} 

If  $(1-\beta)\rho_m\geq \beta^2\rho_n$,  $\frac{(1-\beta)(1+\rho_m|h_m|^2)-\beta}{\beta^2\rho_n}\geq |h_m|^2$ always holds. Applying the joint pdf of $h_m$ and $h_n$ and also following steps similar to those in the proof for Lemma \ref{lemmapn}, the first part of the corollary can be obtained.  

  If $(1-\beta)\rho_m< \beta^2\rho_n$,  whether  $\frac{(1-\beta)(1+\rho_m|h_m|^2)-\beta}{\beta^2\rho_n}\geq |h_m|^2$ holds is depending on the value of $|h_m|^2$. Particular, if $|h_m|^2<\kappa_1$, $\frac{(1-\beta)(1+\rho_m|h_m|^2)-\beta}{\beta^2\rho_n}\geq |h_m|^2$ holds,  otherwise $\frac{(1-\beta)(1+\rho_m|h_m|^2)-\beta}{\beta^2\rho_n}< |h_m|^2$. Hence the probability $\tilde{\mathrm{P}}_n$ can be rewritten as follows: 
\begin{align}
\tilde{\mathrm{P}}_n  = &\mathrm{P} \left(  |h_m|^2\leq \kappa_1\right)\\\nonumber &-\mathrm{P} \left( |h_m|^2\leq \kappa_1,  |h_n|^2\leq \frac{(1-\beta)(1+\rho_m|h_m|^2)-\beta}{\beta^2\rho_n}\right).\end{align}
   
Following steps similar to those in the proof for Lemma \ref{lemmapn},  the second part of the corollary can also  be obtained and hence the proof for the corollary is complete.  
\end{proof}
 
{\it Remark 5:} It is desirable   to have $\tilde{\mathrm{P}}_n\rightarrow 0 $ which means that NOMA-MEC can deliver more data while using less energy compared to OMA-MEC. However,  the  asymptotic property   of $\tilde{\mathrm{P}}_n$ is     depending on whether $(1-\beta)\rho_m< \beta^2\rho_n$ holds. 

\begin{itemize}
\item For the case   $\rho_m$ is a constant and $\rho_n\rightarrow \infty$, we have $(1-\beta)\rho_m< \beta^2\rho_n$.  In this case, $\tilde{\mathrm{P}}_n $ approaches zero,   since  
\begin{align}
\tilde{\mathrm{P}}_n  \leq &1-\frac{M!\sum^{m-1}_{l=0}c_l\frac{e^{-(M-m+l+1)\kappa_1 }}{M-m+l+1}}{(m-1)!(M-m)!}   \\\nonumber 
\underset{(10)}{\approx} &\frac{M!}{(M-m)!}\frac{(1-2\beta)^m}{\beta^{2m}\rho_n^m} \rightarrow 0,
\end{align}
where step (10) follows from steps similar to those in the proof for Lemma \ref{lemmahigh}. For the case that both $\rho_m$ and $\rho_n$ approach infinity and $\frac{\rho_m}{\rho_n}<\frac{\beta^2}{1-\beta}$, the same conclusion can be obtained.

\item  For the case that both $\rho_m$ and $\rho_n$ approach infinity and $\frac{\rho_m}{\rho_n}\geq \frac{\beta^2}{1-\beta}$,  $\tilde{\mathrm{P}}_n $ approaches a non-zero constant, since    
 \begin{align}\nonumber 
\mathrm{P}_n  \rightarrow   &1-c_{mn}\sum^{n-1-m}_{p=0}c_p \sum^{m-1}_{l=0}c_l  \frac{  1}{(M-m-p)\tilde{a} }. 
\end{align}

By applying Lemma \ref{lemma2}, the following holds  
\begin{align}
 \sum^{m-1}_{l=0}c_l  \frac{ 1}{ (\tilde{b}+p+l+1)}  =\frac{m!}{(\tilde{b}+p+1) \cdots (\tilde{b}+p+m)},
 \end{align}
 where $\tilde{b}=\tilde{a}-p-l-1$. 
 Therefore, at high SNR, the probability can be approximated as follows:
 \begin{align}\nonumber 
\mathrm{P}_n  \rightarrow  &1-\sum^{n-1-m}_{p=0}  \frac{m!c_pc_{mn}}{(M-m-p)\prod^{m}_{i=1}(\tilde{b}+p+i) },
 \end{align}
 which is a non-zero constant and not a function of the SNR. 
  If $\rho_n$ is a constant, we will have  $(1-\beta)\rho_m\geq  \beta^2\rho_n$ when $\rho_m\rightarrow \infty$. In this case, $\tilde{\mathrm{P}}_n $ also approaches a non-zero constant.  
  \end{itemize}
  
\section{The Application of  NOMA Downlink Transmission to MEC}\label{section downlink}
This section is to consider another type of MEC scenarios with $M=1$ and $L=K$, i.e.,  a single user has $K$ tasks to be offloaded to $K$ MEC servers. Assume that   the MEC servers   are ordered as follows:
\begin{align}\label{order2}
|g_1|^2\leq \cdots \leq |g_K|^2,
\end{align}
where $|g_m|^2$ denotes the channel gain between the user and MEC server $m$. If OMA is used, the user uses $K$ dedicated time slots with $T$ seconds each to offload its tasks to the servers individually. By using NOMA downlink transmission, the user can offload multiple tasks to multiple servers simultaneously.  Similar to the previous section, we assume that   two  MEC servers, server $m$ and server $n$, are scheduled to perform NOMA.

\subsection{Impact of NOMA on Offloading Latency}
By imposing the constraint that  the user  offloads the task intended to MEC server $n$ within the time slot which would be solely occupied by  server $m$ in the OMA mode, the overall offloading latency  can be significantly reduced. Particularly, in NOMA, the numbers of bits transmitted to the two MEC servers within one time slot  are given by
\begin{align}\label{nm1}
 {N}_m^{NOMA} = T\log\left(1+\frac{P^{ow}\alpha_m^2 |g_m|^2}{P_N+P^{ow} \alpha_n^2|g_m|^2}\right), 
\end{align}
and
\begin{align}\label{nm2}
 {N}_n^{NOMA} =& T\log\left(1+ \frac{P^{ow}}{P_N} \alpha_n^2|g_n|^2\right), 
\end{align}
where $ \alpha_m$ and $\alpha_n$   denote the NOMA power allocation coefficients, and $P^{ow}$ denotes the user's transmit power\footnote{For notational simplicity, subscript $i$ is omitted since there is a single user  ($M=1$).}.        

Therefore, the probability  for   the user to finish offloading its tasks to the   MEC servers    can be expressed as follows: 
\begin{align}
\mathrm{P}^D_m &= \mathrm{P}\left( T\log\left(1+\frac{P^{ow}\alpha_m^2 |g_m|^2}{1+P^{ow} \alpha_n^2|g_n|^2}\right)\geq N_{m} \right), 
\end{align}
and 
\begin{align}
\mathrm{P}^D_n &= \mathrm{P}\left( T \log\left(1+ P^{ow} \alpha_n^2|g_n|^2\right)\geq N_{n} \right),
\end{align}
where the index for the user is omitted, i.e.,  $N_{i,l}$ is simplified as $N_l$.  
When $T$ and $N_{l}$ are fixed, the above probabilities  can be obtained straightforwardly from the existing literature of NOMA~\cite{Zhiguo_CRconoma}.   

\subsection{Impact of NOMA on Offloading Energy Consumption}
 Similar to Section \ref{subsection energy}, a modified NOMA-MEC scheme is considered by using two time slots. During the first time slot, the user  offloads one task to server $m$ and parts of a task to server $n$ simultaneously.  The second time slot is dedicated for the user to offload the remaining parts of the task intended to server $n$. In OMA, the user offloads the two tasks   in   two time slots separately.  Denote the overall energy consumption  in the OMA and NOMA modes by $ {E}^{OMA}$ and $ {E}^{NOMA}$, respectively.    

It is important to point out that the use of NOMA downlink brings a change to the expressions of the offloading  rates. On the one hand,  in OMA, the numbers of bits transmitted to the two MEC servers are given by
\begin{align}
\tilde{N}_i^{OMA} = T\log(1+ \rho |g_i|^2), 
\end{align}
for $i\in\{m,n\}$, 
where $\rho=\frac{{P}^{ow} }{P_N}$,   and it is assumed that the user uses the same transmit power during the two equal-length time slots.  Since $T$ seconds are used, the overall energy consumed in  the two time slots in OMA is $E^{OMA}=2TP^{ow}$ . 

On the other hand,  in NOMA, it is assumed that during the first time slot, the user uses the same transmit power as in the   OMA mode, and uses $\tilde{\beta} P^{ow}$ as the transmit power during the second time slot, where $\tilde{\beta}$ denotes a parameter for  the energy reduction. Therefore, in NOMA, the numbers of bits transmitted to the two MEC servers are  given by \cite{Nomading}
\begin{align}
\tilde{N}_m^{NOMA} =  T\log\left(1+\frac{\rho \alpha_m^2 |g_m|^2}{1+\rho \alpha_n^2|g_n|^2}\right), 
\end{align}
and
\begin{align}
\tilde{N}_n^{NOMA} =& T\log\left(1+ \rho \alpha_n^2|g_n|^2\right) +T\log(1+\tilde{\beta} \rho |g_n|^2) ,
\end{align}
respectively. 
   Since $2T$ seconds are used, the overall consumed energy is  $E^{NOMA}= (1+\tilde{\beta})TP^{ow} $ which implies 
\begin{align}\label{energy}
E^{NOMA}< E^{OMA},
\end{align}
 if $\tilde{\beta} < 1$. 

To ensure that   server  $m$ is connected in the NOMA mode with the same reliability as in   OMA,   the cognitive ratio power allocation policy is used as follows: \cite{Zhiguo_CRconoma}
\begin{align}\label{cr}
 \alpha_n^2 = \max\left\{0, \frac{\rho |g_m|^2 -\epsilon}{\rho |g_m|^2(1+\epsilon)}\right\},
\end{align}
where it is assumed that the user's tasks contain the same number of bits, i.e., $N_{m}=N_{n}\triangleq N$ and $\epsilon=2^{\frac{N}{T}}-1$.  Since  MEC server $m$ experiences the same reliability in the OMA and NOMA modes,  we will only focus on the performance of server $n$   in the following. 

With a choice of $\tilde{\beta}$ satisfying \eqref{energy}, NOMA-MEC uses less energy than OMA-MEC, but it is not guaranteed that NOMA can deliver the same amount of data as OMA-MEC, which is measured by the following probability:  
\begin{align}
\tilde{\mathrm{P}}^D_n\triangleq &\mathrm{P} \left( T\log\left(1+\alpha_n^2\rho|g_n|^2\right)\right.\\\nonumber&+\left.T\log(1+ \tilde{\beta}\rho |g_n|^2) \leq T\log(1+\rho |g_n|^2) \right).
\end{align}
The following lemma provides the closed-form expression for $\tilde{\mathrm{P}}^D_n$.

\begin{lemma}\label{lemmaqmn}
For a fixed choice of $\tilde{\beta}$, the probability for OMA-MEC to deliver more data to server $n$ than NOMA-MEC can be approximated as follows: 
\begin{align} \label{xccc}
\tilde{\mathrm{P}}^D_n= & 1-\frac{K!\sum^{m-1}_{l=0}c_l\frac{e^{-(K-m+l+1)\frac{\epsilon}{\rho} }}{K-m+l+1}}{(m-1)!(K-m)!} \\\nonumber &+c_{mn} \sum^{n-1-m}_{p=0}c_p\sum^{N}_{i=1} \frac{\pi}{N} \left(\frac{\epsilon}{2\tilde{\beta}\rho}- \frac{\epsilon}{2\rho}\right)\\\nonumber &\times f\left(\left(\frac{\epsilon}{2\tilde{\beta}\rho}+ \frac{\epsilon}{2\rho}\right) +\left(\frac{\epsilon}{2\tilde{\beta}\rho}- \frac{\epsilon}{2\rho}\right)\theta_i \right) \sqrt{1-\theta_i^2} ,
\end{align}
where $\theta_i=\cos\left(\frac{2i-1}{2N}\pi\right)$, $N$ denotes the Chebyshev-Gauss approximation parameter, and
\begin{align}
f(x) =& e^{-(p+1)x} (1-e^{-x})^{m-1}\\ \nonumber&\times   \frac{e^{-(K-m-p)x}-e^{-(K-m-p)\frac{\rho x[(1-\tilde{\beta})(1+\epsilon)-1]+\epsilon}{\rho\tilde{\beta}(\rho x-\epsilon)} }}{K-m-p}.
\end{align}
\end{lemma}
\begin{proof}
Please refer to Appendix E. 
 \end{proof}
 
 The high SNR behavior of $\tilde{\mathrm{P}}^D_n$ is shown  in the following lemma.
 \begin{lemma}\label{lemma5}
 At high SNR, i.e., $\rho\rightarrow \infty$, $\tilde{\mathrm{P}}^D_n$ can be approximated as follows:
 \begin{align}
\tilde{\mathrm{P}}^D_n \doteq \rho^{-m},
 \end{align}
 where $\tilde{f}(\rho)\doteq \rho^{-d}$ denotes the exponential equality, i.e., $d= - \underset{\rho\rightarrow \infty}{\lim}\frac{\log\tilde{f}(\rho)}{\log \rho}$ \cite{Zhengl03}. 
 \end{lemma}
 \begin{proof}
 Please refer to Appendix F. 
\end{proof}
{\it Remark 6:} Lemma \ref{lemma5} shows that at high SNR, the probability for NOMA-MEC to outperform OMA-MEC becomes one, which can be explained   in the following. Recall that the use of the cognitive radio power allocation policy is to satisfy server $m$'s requirements  before allocating any power to server $n$. At high SNR, more power becomes available to server $n$, which means that a significant amount of data can be offloaded to server $n$ during the first time slot, and hence the overall amount of the offloaded data over the two time slots  is also improved.    

{\it Remark 7:} Lemma \ref{lemma5} also indicates that MEC server $m$'s channel condition has a critical impact on the probability $\tilde{\mathrm{P}}_n^D$. In particular, scheduling a server with better channel conditions to act as server $m$ improves the probability that NOMA-MEC outperforms OMA-MEC.  It is worth pointing out that, for the application of NOMA uplink transmission to MEC, a different conclusion was made in Lemma \ref{lemmahigh} which states that scheduling a user with poor channel conditions is beneficial to the implementation of NOMA.

\section{Simulation Results and Discussions}
In this section, the performance of NOMA-MEC is evaluated by using computer simulations, where the accuracy of the developed analytical results is also verified. 

The impact of NOMA uplink transmission on MEC is examined first. Recall that  the NOMA-MEC schemes described in Section \ref{section uplink} ensure that user $n$ is served without causing any performance degradation to user $m$, so only user $n$'s performance is focused. 
In Fig. \ref{fig01x}, the offloading probability $\mathrm{P}_n$ is shown as a function of user $n$'s transmit power. Note that    the noise power is assumed to be normalized, which means that    user $n$'s transmit power  is the same as $\rho_n$.   Fig. \ref{fig01x} shows that the behavior of $\mathrm{P}_n$ is depending on the relationship between the two users' transmit powers. 
When user $m$'s transmit power ($\rho_m$) is fixed, Fig. \ref{fig 1 b} demonstrates that increasing user $n$'s transmit power can increase $\mathrm{P}_n$. This phenomenon can be explained in the following. When $\rho_m$ is fixed, the time duration required by user $m$ to offload its task, $T_m$, is also fixed. On the other hand, increasing $\rho_n$   increases  user $n$'s offloading data rate, which makes it more likely for user $n$ to complete its offloading within the fixed time duration $T_m$. 

 If both $\rho_m$ and $\rho_n$ approach infinity and  the ratio of the two users' powers is a constant, Fig. \ref{fig 1 a} shows that $\mathrm{P}_n$ goes to zero. This phenomenon is due to the fact that increasing $\rho_m$ reduces $T_m$, the time duration required by user $m$ to complete its offloading. On the other hand, recall that  $\mathrm{P}_n$ measures  the likelihood  for user $n$ to complete its offloading by only using $T_m$, the time slot which would be solely occupied by user $m$ in the OMA mode. Therefore, reducing $T_m$ means that there is less opportunity for user $n$ to use NOMA for offloading, which leads to the reduction of  $\mathrm{P}_n$. It is worth pointing out that the two subfigures in Fig. \ref{fig01x} show that the curves for the analytical results perfectly match the ones for the simulation results, which verifies the accuracy of the developed analytical results.  

\begin{figure}[!htp] 
\begin{center}\subfigure[$\rho_m=10$ dB ]{\label{fig 1 b}\includegraphics[width=0.45\textwidth]{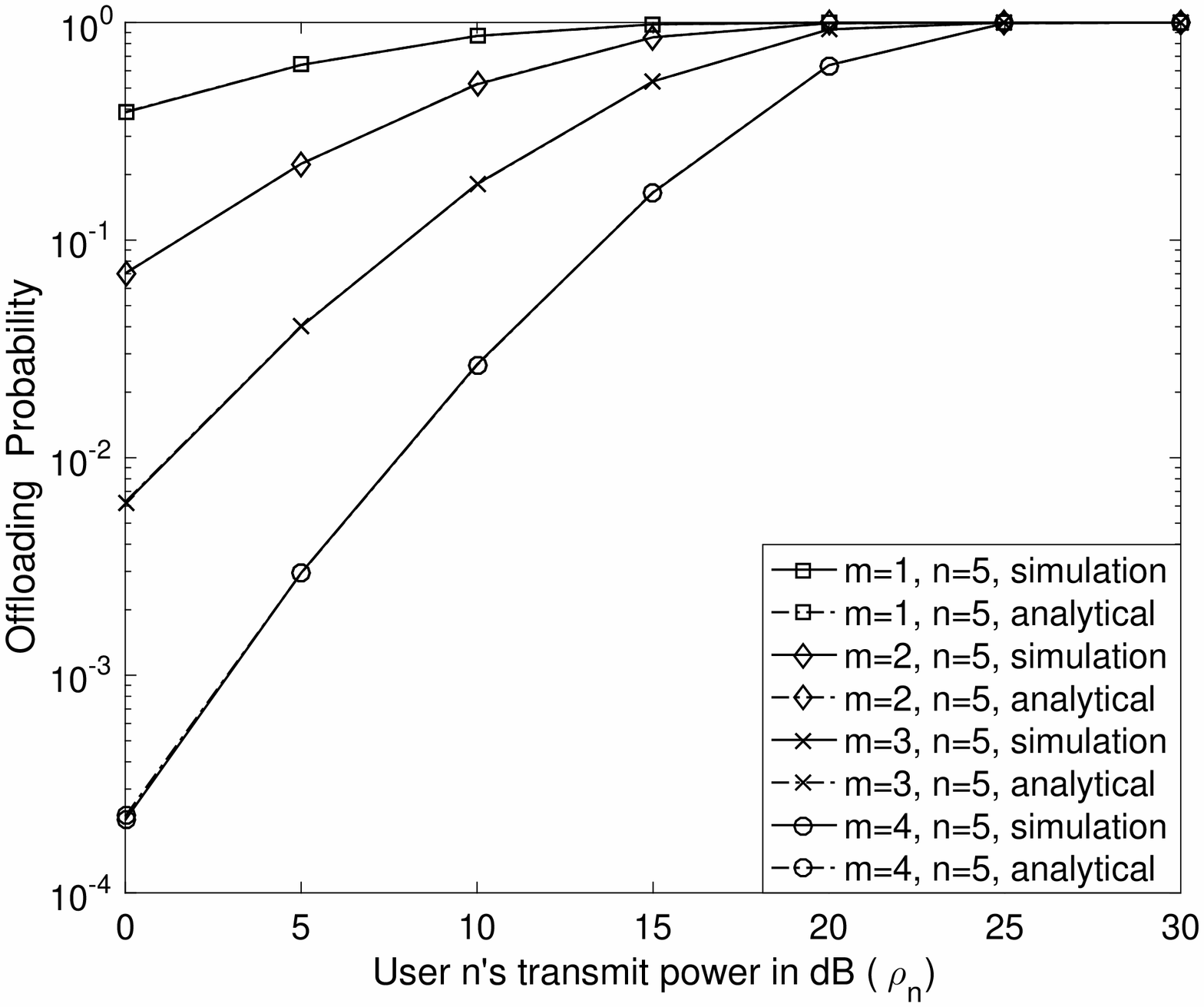}} \subfigure[ $\eta=\frac{\rho_n}{\rho_m}=2$ ]{\label{fig 1 a}\includegraphics[width=0.45\textwidth]{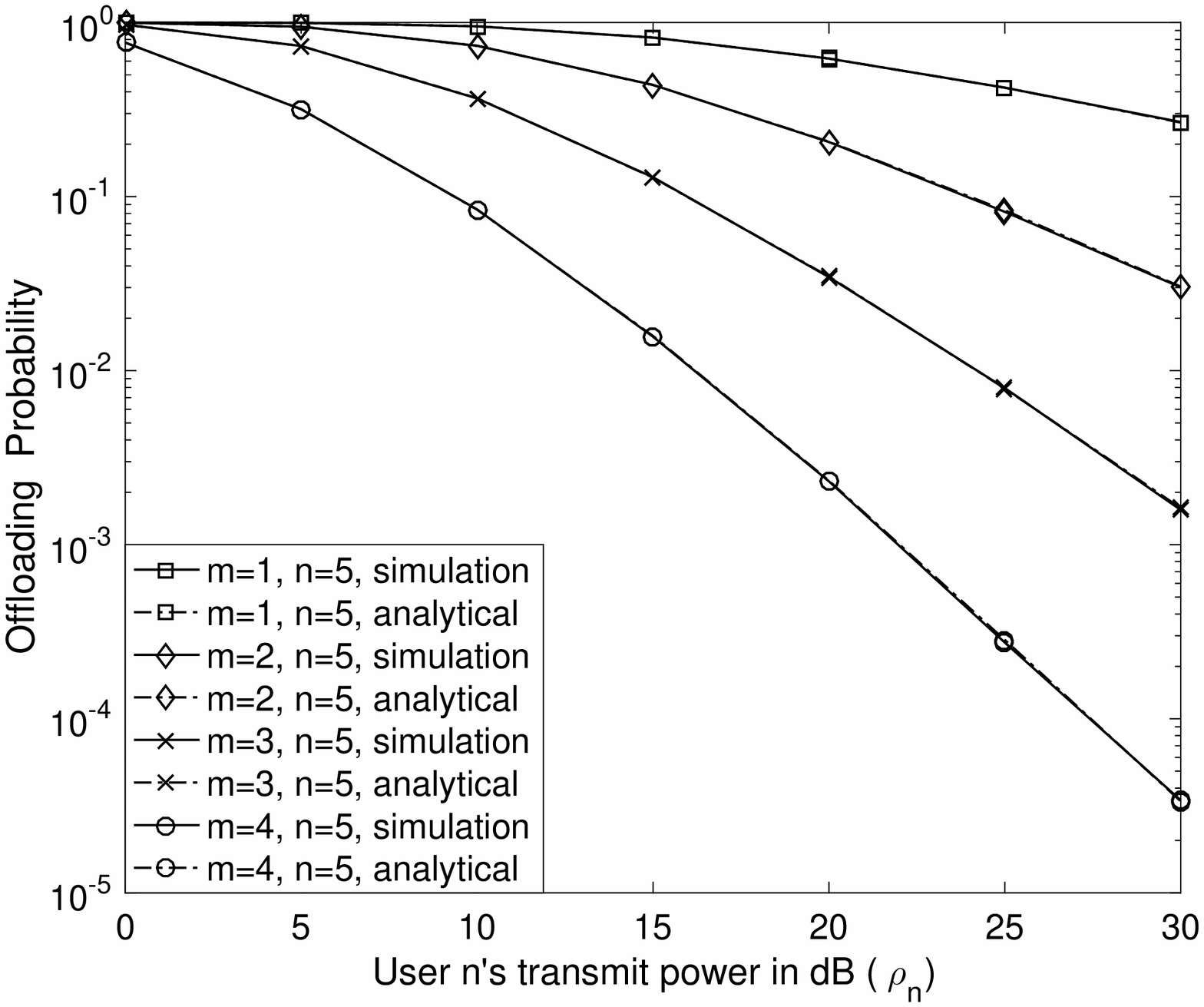}}
\vspace{-1em}
\end{center}
\caption{  The probability for user $n$ to complete its offloading by using the time slot allocated to user $m$, $\mathrm{P}_n$. There are five users $M=5$.    }\label{fig01x}\vspace{-0.5em}
\end{figure}

\begin{figure}[!htbp]\centering 
    \epsfig{file=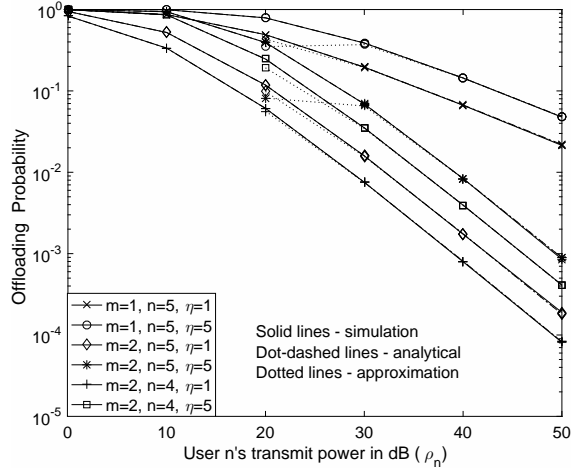, width=0.45\textwidth, clip=}\vspace{-1em}
\caption{ The impact of the parameters, such as $m$, $n$, and $\eta$ on the offloading probability, $\mathrm{P}_n$. There are five users $M=5$ and $\eta=\frac{\rho_n}{\rho_m}$.   \vspace{-1em} }\label{fig2}
\end{figure}

In Fig. \ref{fig2}, the impact of the parameters, such as $m$, $n$, and $\eta$, on the offloading probability $\mathrm{P}_n$ is shown. As pointed out in the remarks      for Lemma \ref{lemmahigh},   the probability $\mathrm{P}_n$ is inversely proportional to $\rho_n^{\frac{m}{2}}$. This conclusion is confirmed by Fig. \ref{fig2} as one can observe that the choice of  $m$ has a critical impact on  $\mathrm{P}_n$. On the other hand,  reducing $n$ also reduces the probability, but its impact on   the probability is not as significant as $m$.   For a fixed $\rho_n$,    increasing $\eta$ reduces user $m$'s transmit power, which means that user $m$ needs more time for offloading, i.e., $T_m$ is increased.  Since there is more time available for user $n$ to offload, the offloading probability is improved, as can be observed from Fig. \ref{fig2}.  Furthermore, the high SNR approximation obtained in Lemma \ref{lemmahigh} is also verified in the figure. While this approximation is not accurate in the low SNR regime, it matches the simulation results perfectly at high SNR.

 \begin{figure}[!htp]\vspace{-1em}
\begin{center} \subfigure[$\rho_m=10$ dB ]{\label{fig 31 a}\includegraphics[width=0.45\textwidth]{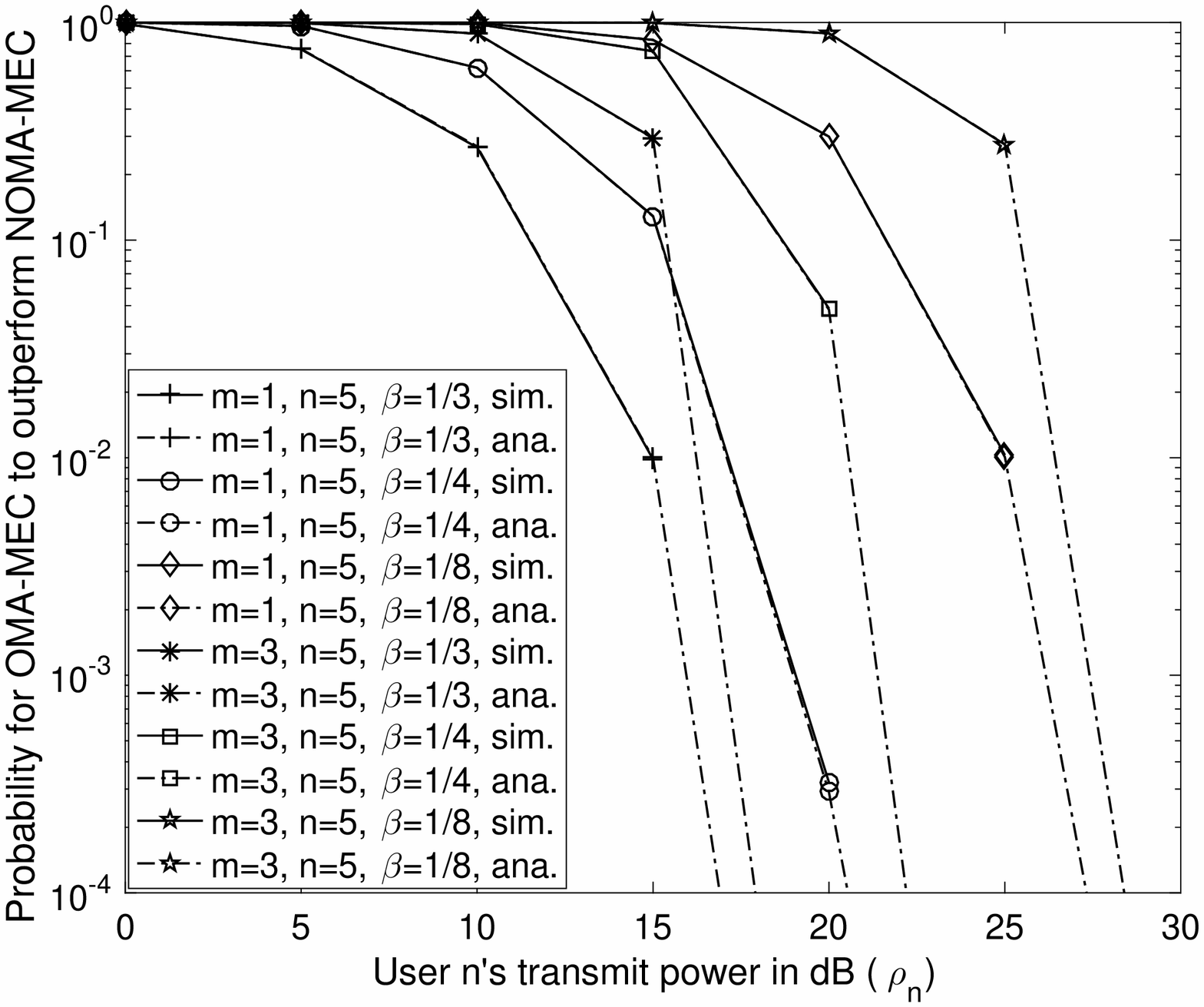}}
\subfigure[$\eta=\frac{\rho_n}{\rho_m}=5$ ]{\label{fig 32 b}\includegraphics[width=0.45\textwidth]{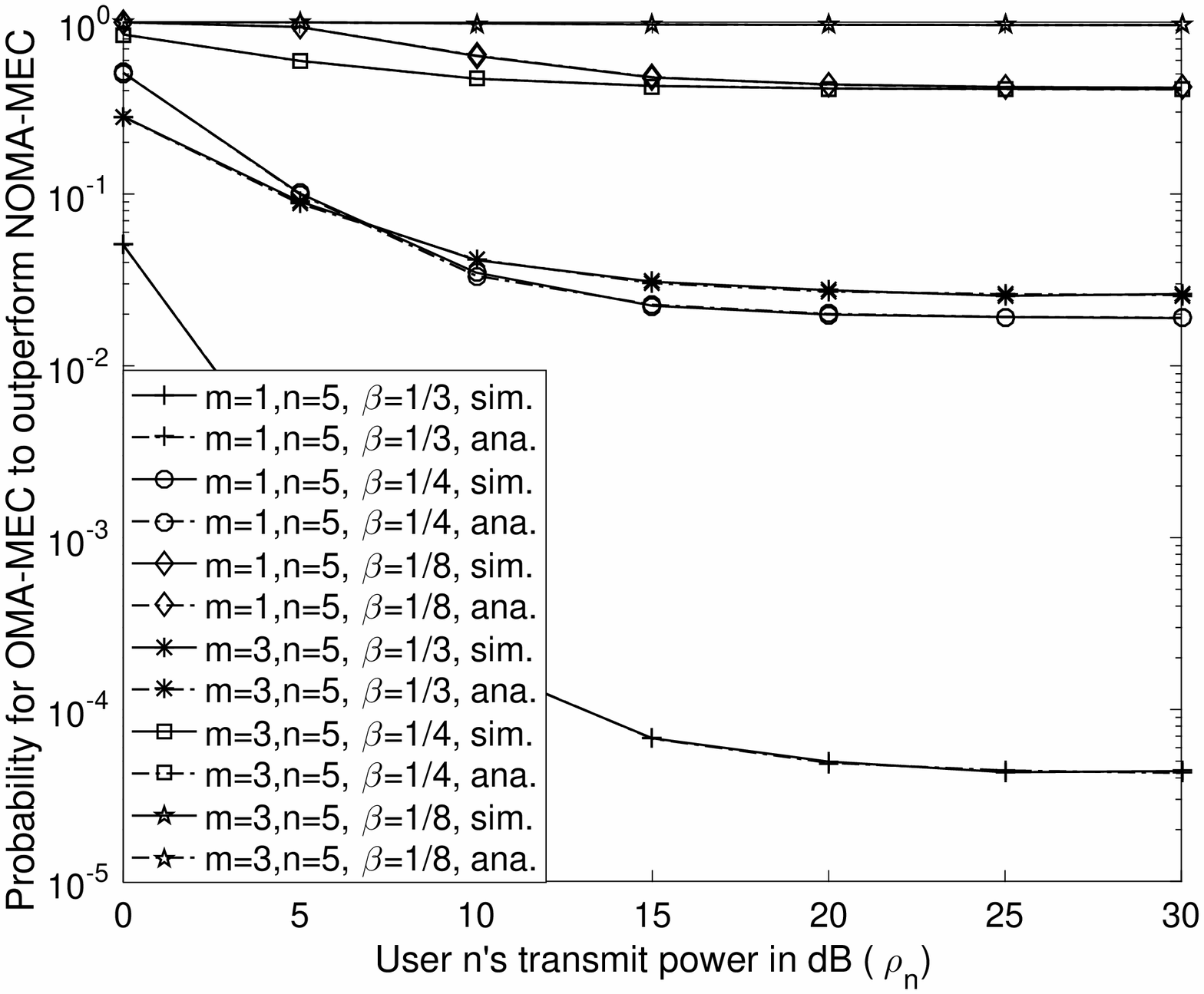}}\vspace{-1em}
\end{center}
\caption{ The probability for OMA-MEC to outperform NOMA-MEC, $\tilde{\mathrm{P}}_n$.  There are five users $M=5$.   \vspace{-1em} }\label{fig3}
\end{figure}

The impact of NOMA-MEC on the energy consumption is examined in Fig. \ref{fig3}. As can be observed from the figure, the use of NOMA can significantly reduce user $n$'s energy consumption for offloading.  In particular, first recall from \eqref{power} that the ratio between the energy consumption in the OMA and NOMA modes is $2\beta$. 
As shown in Fig. \ref{fig 31 a},  if the energy used by NOMA-MEC  is  only a quarter  of the energy used by  OMA-MEC, i.e., $\beta=\frac{1}{8}$, the probability for OMA-MEC to outperform NOMA-MEC, $\tilde{\mathrm{P}}_n$, can be reduced to  $10^{-2}$ when $\rho_n=25$ dB and $m=1$.  If the energy  of NOMA-MEC is just half of the energy  used in the OMA mode, it becomes almost sure that NOMA-MEC outperforms OMA-MEC, after $\rho_n$ is larger than $15$ dB. Recall that Remark 5 points out that for the case that both $\rho_m$ and $\rho_n$ approach infinity and $\frac{\rho_m}{\rho_n}\geq \frac{\beta^2}{1-\beta}$,   $\tilde{\mathrm{P}}_n $ approaches a non-zero constant, which is confirmed by Fig. \ref{fig 32 b}.  It is worth pointing out that user pairing has a significant impact on   energy saving of NOMA-MEC, as can be seen from the figure. For example, in Fig. \ref{fig 31 a}, when  $\rho_n=15$ dB,   the case with $m=1$ and $\beta=\frac{1}{4}$ can even realize a smaller $\tilde{\mathrm{P}}_n$  than the case with $m=3$ and $\beta=\frac{1}{3}$, i.e., scheduling user $1$ as the NOMA weak user can save more energy than the case of $m=3$.  Note that  the subfigures in Fig. \ref{fig3} also demonstrate the accuracy of the analytical results developed in Corollary~\ref{corollary1}.

 \begin{figure}[!htp]\vspace{-1em}
\begin{center} \subfigure[$\tilde{\beta}=\frac{1}{2}$ ]{\label{fig 12 a}\includegraphics[width=0.45\textwidth]{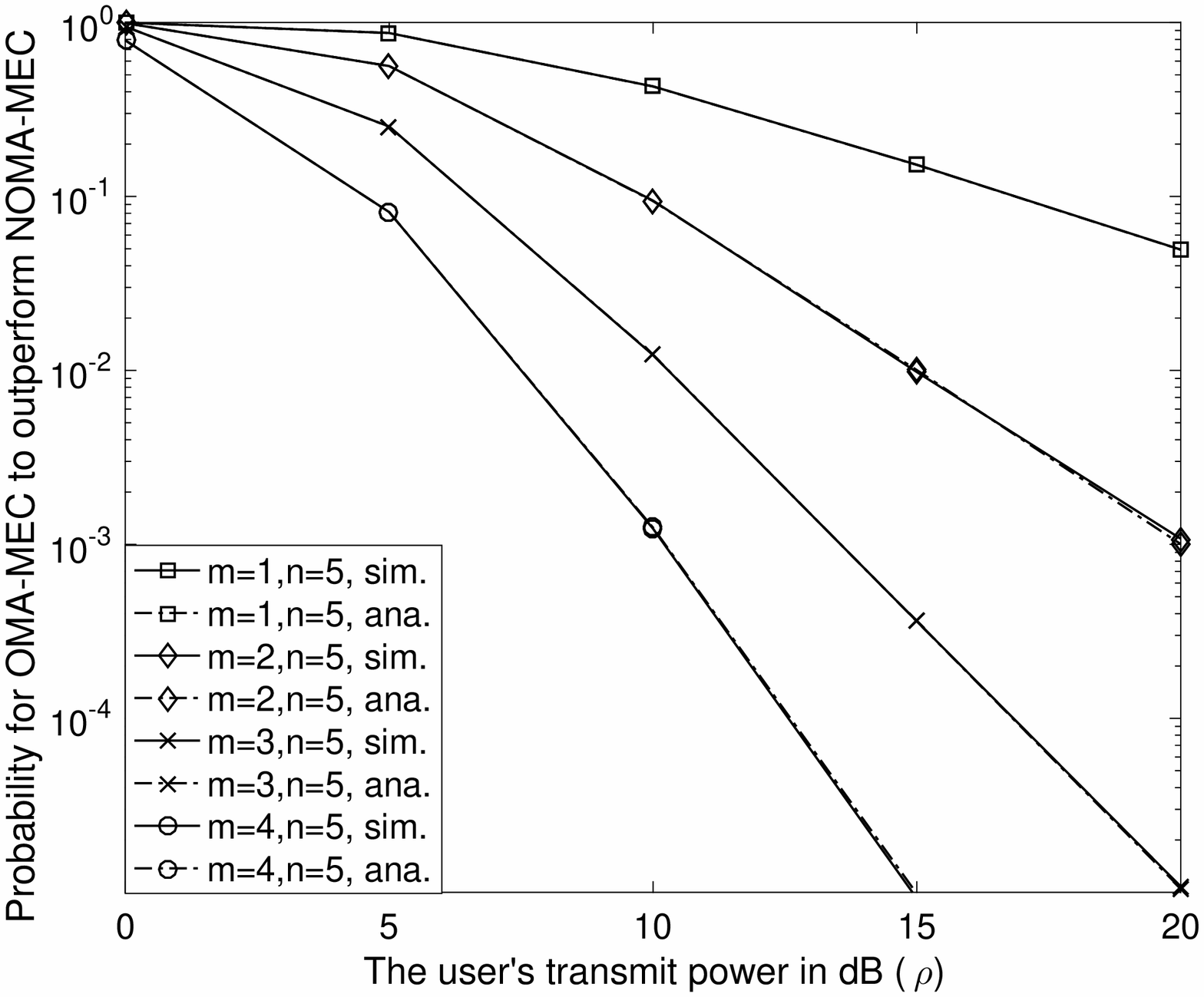}}
\subfigure[$\tilde{\beta}=\frac{1}{5}$ ]{\label{fig 12 b}\includegraphics[width=0.45\textwidth]{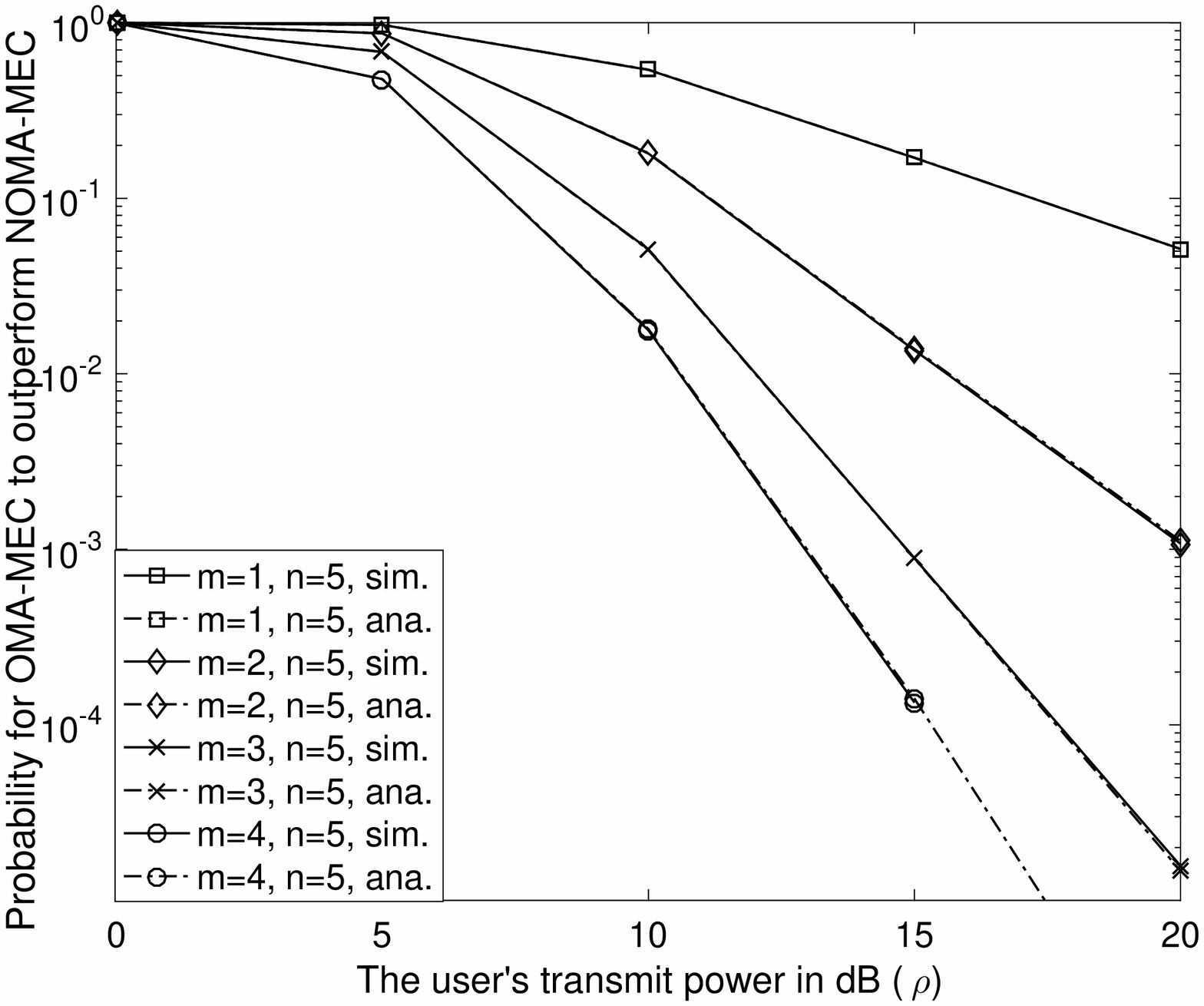}}\vspace{-1em}
\end{center}
\caption{  The probability for OMA-MEC to outperform NOMA-MEC, $\tilde{\mathrm{P}}^D_n$.  There are five users $M=5$.     }\label{fig04x}\vspace{-0.5em}
\end{figure}

In Fig. \ref{fig04x}, the impact of NOMA downlink transmission on MEC is illustrated. Because   the cognitive radio power allocation policy is used, user $m$'s performance is not affected even though user $n$ is admitted to the time slot which would be solely occupied by user $m$ in the OMA mode.  Therefore, only user $n$'s performance is evaluated. As can be observed from Fig. \ref{fig04x}, the probability for OMA-MEC to outperform NOMA-MEC approaches zero by increasing the transmit power. This phenomenon is due to the   fact that,   at high SNR, more power becomes  available to user $n$ for its offloading. One can also observe that   the slope of the probability, $\tilde{\mathrm{P}}_n^D$, is determined by the choice of $m$. This observation is consistent to Lemma \ref{lemma5}, which states that the decay rate of $\tilde{\mathrm{P}}_n^D$ is depending on $m$. 

It is interesting to point  out that the effects  of $m$ in different NOMA-MEC scenarios are different. Particular, for the MEC scenario considered in Fig. \ref{fig01x},   increasing $m$ degrades the performance of NOMA-MEC, but for the scenario considered in Fig. \ref{fig04x}, increasing $m$ improves the performance of NOMA-MEC. The reason for the two different effects is explained in the following. For the scenario considered in Fig. \ref{fig01x}, i.e., the application of NOMA uplink transmission to MEC, increasing $m$, i.e., scheduling a user with better channel conditions to act as the NOMA weak user, reduces $T_m$, the offloading time  required by user $m$. Therefore,   it is less likely for user $n$ to offload its task to the server   within the shortened time interval $T_m$. In the scenario considered in Fig. \ref{fig04x}, i.e., the application of NOMA downlink transmission to MEC, increasing $m$, i.e., scheduling a server with better channel conditions to act as the NOMA weak user, reduces the power consumed by  server $m$, and hence there is more power available to perform NOMA and offload the user's task to server $n$. 

\section{Conclusions}
In this paper, we have investigated the   coexistence of NOMA and MEC. The application of NOMA uplink transmission to MEC was considered first, where the use of NOMA ensures that multiple users can perform offloading at the same time. Then, the application of NOMA downlink transmission to MEC was studied, where one user uses NOMA to offload multiple tasks to multiple MEC servers simultaneously.  Analytical results have been developed to demonstrate that the use of NOMA can efficiently reduce  the latency and energy consumption of MEC offloading. In addition, various asymptotic studies have also been carried out to reveal the impact of the users' channel conditions and transmit powers on the performance of the combined NOMA and MEC system.

\appendices 
\section{Proof for Lemma \ref{lemmapn}}

Recall that $\mathrm{P}_n$ can be rewritten as follows:
\begin{align}
\mathrm{P}_n &= \mathrm{P}\left( R_nT_m\geq N \right)
\\\nonumber
&= \mathrm{P}\left(  \log\left(1+\frac{\rho_n|h_n|^2}{\rho_m|h_m|^2+1}\right) > \log(1+\rho_m|h_m|^2)  \right).
\end{align}
With some algebraic manipulations, the probability can be further  rewritten as follows: 
\begin{align}
\mathrm{P}_n   
&= \mathrm{P}\left(  |h_n|^2  > \frac{\rho_m}{\rho_n} |h_m|^{2} +\frac{\rho_m^2}{\rho_n}|h_m|^{4} \right). 
\end{align} 
Recall that there is an implicit constraint,  $|h_n|^2\geq |h_m|^2$, which leads to  the following inequality:
\begin{align}
 \frac{\rho_m}{\rho_n} |h_m|^{2} +\frac{\rho_m^2}{\rho_n}|h_m|^{4} \geq  |h_m|^{2}.
\end{align} 

Due to the space limits, we only focus on the case with $\rho_n\geq \rho_m$, where the results for the case with $\rho_n<\rho_m$ can be obtained similarly.  $\mathrm{P}_n   $ can be expressed as the sum of the two following probabilities:
\begin{align}\nonumber 
\mathrm{P}_n   
=& \underset{T_1}{\underbrace{\mathrm{P}\left(  |h_n|^2  > \frac{\rho_m}{\rho_n} |h_m|^{2} +\frac{\rho_m^2}{\rho_n}|h_m|^{4},   |h_m|^{2}\geq \frac{\rho_n-\rho_m}{\rho^2_m}\right)}}\\  & \label{t1}
+\underset{T_2}{\underbrace{\mathrm{P}\left(  |h_n|^2  >  |h_m|^{2},   |h_m|^{2}< \frac{\rho_n-\rho_m}{\rho^2_m}\right)}}. 
\end{align} 

By using the order statistics, the joint pdf of $|h_m|^2$ and $|h_n|^2$ can be written as follows: \cite{David03}
\begin{align}\label{joint}
f_{|h_m|^2, |h_n|^2}(x,y) = &c_{mn} e^{-x}e^{-(M-n+1)y}(1-e^{-x})^{m-1} (e^{-x}-e^{-y})^{n-1-m}, 
\end{align}
where $c_{mn}$ is defined in the lemma.

Therefore, the probability $T_1$ can be expressed as follows:
\begin{align}
T_1
=& 
c_{mn} \sum^{n-1-m}_{p=0}c_p\int^{\infty}_{\frac{\rho_n-\rho_m}{\rho^2_m}}e^{-(p+1)x}(1-e^{-x})^{m-1}   \frac{1}{M-m-p}   e^{-(M-m-p)(\frac{\rho_m}{\rho_n} x +\frac{\rho_m^2}{\rho_n}x^{2})} dx,
\end{align}
where $c_p$ is defined in the lemma.

By applying the binomial expansion, the probability $T_1$ can be expressed as follows:
\begin{align}
T_1
=& 
c_{mn} \sum^{n-1-m}_{p=0} \frac{c_p}{M-m-p}\sum^{m-1}_{l=0}\int^{\infty}_{\frac{\rho_n-\rho_m}{\rho^2_m}} c_l   e^{-\left(p+l+1+(M-m-p)\frac{\rho_m}{\rho_n}\right)x - (M-m-p) \frac{\rho_m^2}{\rho_n}x^2} dx.
\end{align}
To make the probability integral applicable, the probability can be further expressed as follows: 
\begin{align}\label{t2}
T_1
=& 
c_{mn} \sum^{n-1-m}_{p=0} \frac{c_p}{M-m-p}\sum^{m-1}_{l=0}\int^{\infty}_{\frac{\rho_n-\rho_m}{\rho^2_m}} c_l     e^{-a\left(x+\frac{b}{2a}\right)^2 +\frac{b^2}{4a}} dx.
\end{align}
By applying Eq. (3.321.2) in \cite{GRADSHTEYN} to the above equation,  the first part in \eqref{eqlemma1} is proved.

Again applying the joint pdf in \eqref{joint},  the probability $T_2$ can be obtained  as follows:
 \begin{align}\label{t3}
T_2= &1-\frac{M!}{(m-1)!(M-m)!} \sum^{m-1}_{l=0}c_l\frac{e^{-(M-m+l+1)\frac{\rho_n-\rho_m}{\rho^2_m} }}{M-m+l+1} .
\end{align}
By substituting \eqref{t2} and \eqref{t3} into \eqref{t1}, the proof for the lemma is complete. 
  
  \section{Proof for Proposition \ref{lemma2}}
  
We first rewrite the sum of the binomial coefficients as follows:
\begin{align}
&\frac{M!}{(m-1)!(M-m)!} \sum^{m-1}_{l=0}c_l \frac{
1 }{ (M-m+l+1)}\\\nonumber =&
\frac{M!}{(m-1)!(M-m)! (M-m+1)} \sum^{m-1}_{l=0}c_l \frac{
1 }{ (1+\frac{l}{M-m+1})} .
\end{align}

From \cite{8008769d}, we can have the following property for the binomial coefficients: 
\begin{align}
\sum^{t}_{i=0} (-1)^{i}{t \choose i}\frac{1}{1+ix} = \frac{t!x^t}{\prod^{t}_{i=1} (1+ix)},
\end{align}
for $x>0$ and $t$ is a non-negative integer. 

By letting $x=\frac{1}{M-m+1}$ and $t=m-1$, the above property can be rewritten as follows:
\begin{align}
&\frac{M!}{(m-1)!(M-m)!} \sum^{m-1}_{l=0}c_l \frac{
1 }{ (M-m+l+1)}\\\nonumber =&  
\frac{M!}{(m-1)!(M-m)! (M-m+1)} \frac{(m-1)!\frac{1}{(M-m+1)^{m-1}}}{\prod^{m-1}_{i=1} (1+\frac{i}{M-m+1})},
\end{align}
which  can be further simplified as follows:
\begin{align}
&\frac{M!}{(m-1)!(M-m)!} \sum^{m-1}_{l=0}c_l \frac{
1 }{ (M-m+l+1)}\\\nonumber =&  
\frac{M!}{ (M-m)! (M-m+1)} \frac{ 1}{\prod^{m-1}_{i=1} (i+M-m+1)} 
\frac{M!}{ (M-m)! (M-m+1)} \frac{ 1}{\prod^{M}_{i=M-m+1} l}=1.
\end{align}
The proof for the proposition is complete. 
\section{Proof for Lemma \ref{lemmahigh}}
Recall that the probability $\mathrm{P}_n   =T_1+T_2$. In the following, the approximation for $T_1$ is obtained first, and then  the approximation for $T_2$ is developed.  

Since both $\rho_m$ and $\rho_n$ approach infinity and $\eta$ is a constant, we can have the following approximation:
\begin{align}\label{maxapp}
\frac{\max\{0,\rho_n-\rho_m\}}{\rho^2_m}+\frac{b\sqrt{a}}{2a} \approx \frac{b\sqrt{a}}{2a},
\end{align}
which implies that whether $\eta\geq 1$ or $\eta<1$ has no impact on the high SNR approximation for $T_1$.  

First recall that the probability integral function $\Phi(x)$ has the following series representation:
\begin{align}\label{pif}
\Phi(x) &=\frac{2}{\sqrt{\pi}}e^{-x^2}\sum^{\infty}_{k=0} \frac{2^kx^{2k+1}}{(2k+1)!!} .
\end{align}

By using the approximation in \eqref{maxapp} and the series representation in \eqref{pif}, the first term of the probability $\mathrm{P}_n   $, $T_1$,  can be rewritten as follows:
\begin{align}
T_1  
\approx  & 
c_{mn} \sum^{n-1-m}_{p=0} \frac{c_p}{M-m-p}\sum^{m-1}_{l=0} c_l  e^{\frac{b^2}{4a}}   \frac{\sqrt{\pi}}{2\sqrt{a}}\left(1 - \frac{2}{\sqrt{\pi}}e^{-\left(\frac{b\sqrt{a}}{2a}\right)^2}\sum^{\infty}_{k=0} \frac{2^k\left(\frac{b\sqrt{a}}{2a}\right)^{2k+1}}{(2k+1)!!}  \right). 
\end{align}
To facilitate the asymptotic studies, the series representation of the exponential functions, $e^{\frac{b^2}{4a}} $, is used and  the probability $T_1   $ can be expressed as follows:
\begin{align}
T_1  \nonumber 
= & 
c_{mn} \sum^{n-1-m}_{p=0} \frac{c_p}{M-m-p} \left( \underset{Q_1}{\underbrace{\sum^{m-1}_{l=0} c_l \frac{\sqrt{\pi}}{2\sqrt{a}}\sum^{\infty}_{s=0} \frac{b^{2s}}{s! 4^sa^s}}} - \underset{Q_2}{\underbrace{ \sum^{m-1}_{l=0} c_l \frac{1}{\sqrt{a}}\sum^{\infty}_{k=0} \frac{2^k\left(\frac{b\sqrt{a}}{2a}\right)^{2k+1}}{(2k+1)!!}}}  \right),
\end{align}
where the two terms, $Q_1$ and $Q_2$, are evaluated separately in the following two subsections. 

\subsection{High SNR Approximation for  $Q_1$}
Recall that $b=l+\lambda$ . 
To facilitate the high SNR approximation,  the binomial expansion is applied to the term $b^{2s}$ and we have the following expression: 
\begin{align}
Q_1 =&  \sum^{\infty}_{s=0} \frac{\sqrt{\pi}}{s! 2^{2s+1}a^{s+\frac{1}{2}}} \sum^{m-1}_{l=0} c_lb^{2s} 
=  \sum^{\infty}_{s=0} \frac{\sqrt{\pi}}{s! 2^{2s+1}a^{s+\frac{1}{2}}} \sum^{m-1}_{l=0} c_l
\sum^{2s}_{q=0} {2s\choose q} \lambda^{2s-q} l^q.
\end{align}

By exchanging the order of the sums,   $Q_1$ can be rewritten as  follows:
\begin{align}  
Q_1  
=&  \sum^{\infty}_{s=0} \frac{\sqrt{\pi}\sum^{2s}_{q=0} {2s\choose q} \lambda^{2s-q}}{s! 2^{2s+1}a^{s+\frac{1}{2}}}
 \sum^{m-1}_{l=0} c_l l^q .
\end{align}
 
\subsubsection{If $m$ is an odd number}  Recall that the following properties  of the binomial coefficients: 
\begin{align}\label{q1}
 \sum^{m-1}_{l=0} c_l  l^t=0,
\end{align}
for $0\leq t\leq m-2$,
and
\begin{align}\label{q2}
 \sum^{m-1}_{l=0} c_l  l^{m-1}=(-1)^{m-1}(m-1)!.
\end{align}
 
Note that when $m$ is an odd number, $(m-1)$ is an even number. In this case, $Q_1  $ can be approximated at high SNR as follows:
\begin{align}\nonumber 
Q_1  
 \approx&  \sum^{\frac{m-1}{2}}_{s=0} \frac{\sqrt{\pi}
}{s! 2^{2s+1}a^{s+\frac{1}{2}}} \sum^{2s}_{q=0} {2s\choose q} \lambda^{2s-q}\sum^{m-1}_{l=0} c_l l^q 
\\ 
\underset{(1)}{\approx}&  \frac{\sqrt{\pi} }{\left(\frac{m-1}{2}\right)! 2^{m}a^{\frac{m}{2}}}
\sum^{m-1}_{q=0} {m-1\choose q} \lambda^{m-1-q}\sum^{m-1}_{l=0} c_l l^q 
  \label{zz2}
\underset{(2)}{=}  \frac{\sqrt{\pi}(-1)^{m-1}(m-1)!,}{\left(\frac{m-1}{2}\right)! 2^{m}a^{\frac{m}{2}}} ,
\end{align}
where  step (1)   follows from \eqref{q1}, and step (2) follows from \eqref{q1} and  \eqref{q2}.  

\subsubsection{if $m$ is an even number} In this case, $(m-1)$ becomes an odd number and $2\left\lceil\frac{m-1}{2}\right\rceil = m$, where  $\lceil\cdot\rceil$ denotes the ceiling function.   Therefore,   $Q_1  $ can be approximated at high SNR as follows:
\begin{align}\nonumber 
Q_1  
 \approx&  \sum^{\left\lceil\frac{m-1}{2}\right\rceil}_{s=0} \frac{\sqrt{\pi}
}{s! 2^{2s+1}a^{s+\frac{1}{2}}} \sum^{2s}_{q=0} {2s\choose q} \lambda^{2s-q}\sum^{m-1}_{l=0} c_l l^q 
\\\nonumber
\underset{(3)}{\approx}&   \frac{\sqrt{\pi}
}{\left\lceil\frac{m-1}{2}\right\rceil! 2^{m+1}a^{\left\lceil\frac{m-1}{2}\right\rceil+\frac{1}{2}}} \sum^{m}_{q=0} {m\choose q} \lambda^{m-q}\sum^{m-1}_{l=0} c_l l^q \\ 
\underset{(4)}{=}&   \frac{\sqrt{\pi}}{\left\lceil\frac{m-1}{2}\right\rceil! 2^{m+1}a^{\left\lceil\frac{m-1}{2}\right\rceil+\frac{1}{2}}}  \left(m \lambda \sum^{m-1}_{l=0} c_l l^{m-1}+ 
  \sum^{m-1}_{l=0} c_l l^m\right),\label{zz1}
\end{align}
where     \eqref{q1} is used to obtain steps (3) and (4). After applying \eqref{q2},  $Q_1$ can be approximated  as follows:
\begin{align}\nonumber 
Q_1  \approx &   \frac{\sqrt{\pi}}{\left(\frac{m}{2}\right)! 2^{m+1}a^{\frac{m+1}{2}}}  \left(m \lambda (-1)^{m-1}(m-1)!+ 
  \sum^{m-1}_{l=0} c_l l^m\right).
\end{align}

\subsection{High SNR Approximation for $Q_2$}
On the other hand, after applying the binomial expansion to $b^{2s}$, $Q_2$ can be expressed as follows:
\begin{align}
Q_2=& \sum^{\infty}_{k=0} \frac{ 1}{(2k+1)!!2^{k+1}a^{k+1}}\sum^{m-1}_{l=0} c_l b^{2k+1}\\\nonumber
=& \sum^{\infty}_{k=0} \frac{ 1}{(2k+1)!!2^{k+1}a^{k+1}}\sum^{2k+1}_{p=0} {2k+1 \choose p}  \lambda^{2k+1-p} \sum^{m-1}_{l=0} c_l l^p.
\end{align}
Depending on the value of $m$, $Q_2$ can be evaluated differently   in the following subsections. 
\subsubsection{if $m$ is an odd number} In this case,  $(m-2)$ is still an odd number, and $2\left\lceil\frac{m-2}{2}\right\rceil =m-1$. Therefore,  $Q_2$ can be written as follows:  
\begin{align}
Q_2 \nonumber
 \approx & \sum^{\left\lceil\frac{m-2}{2}\right\rceil }_{k=0} \frac{ 1}{(2k+1)!!2^{k+1}a^{k+1}}\sum^{2k+1}_{p=0} {2k+1 \choose p}  \lambda^{2k+1-p} \sum^{m-1}_{l=0} c_l l^p
\\\nonumber
\underset{(5)}{\approx}&    \frac{ 1}{m!!2^{\left\lceil\frac{m-2}{2}\right\rceil +1}a^{\left\lceil\frac{m-2}{2}\right\rceil +1}}\sum^{m}_{p=0} {m \choose p}  \lambda^{m-p} \sum^{m-1}_{l=0} c_l l^p
\\\label{xx1}
\underset{(6)}{=}&   \frac{ 1}{m!!2^{\frac{m+1}{2}}a^{\frac{m+1}{2}}}\left(   \sum^{m-1}_{l=0} c_l l^m +  m  \lambda  (-1)^{m-1}(m-1)!
\right),
\end{align}
where  step (5)   follows from \eqref{q1}, and step (6) follows from \eqref{q1} and  \eqref{q2}.  

\subsubsection{if $m$ is an even number} In this case,   $(m-2)$ is also an even number. Following steps similar to those in the previous subsections, $Q_2$ can be evaluated as follows:
\begin{align}
Q_2 \nonumber
 {\approx}& \sum^{\frac{m-2}{2}}_{k=0} \frac{ 1}{(2k+1)!!2^{k+1}a^{k+1}}\sum^{2k+1}_{p=0} {2k+1 \choose p} \lambda^{2k+1-p} \sum^{m-1}_{l=0} c_l l^p
\\\nonumber
 {\approx}&     \frac{ 1}{(m-1)!!2^{\frac{m-2}{2}+1}a^{\frac{m-2}{2}+1}}\sum^{m-1}_{p=0} {m-1 \choose p}  \lambda^{m-p} \sum^{m-1}_{l=0} c_l l^p
\\\label{xx2}
 {=}&      \frac{ 1}{(m-1)!!2^{\frac{m}{2}}a^{\frac{m}{2}}} (-1)^{m-1}(m-1)!.
\end{align}
Combining \eqref{zz2}, \eqref{zz1}, \eqref{xx1} and  \eqref{xx2}, the approximation for $T_1$ can be obtained.  

On the other hand, the approximation for $T_2$ can be obtained by first rewriting $T_2$ as follows: 
 \begin{align}
T_2= &1-\frac{M!}{(m-1)!(M-m)!} \sum^{m-1}_{l=0}c_l \frac{\sum^{\infty}_{k=0}(-1)^k
\frac{(M-m+l+1)^k  \frac{( \eta-1)^k}{\rho_m^k}  }{k!}
}{M-m+l+1}
\\\nonumber
=&1-\frac{M!}{(m-1)!(M-m)!} \sum^{m-1}_{l=0}c_l 
 (M-m+l+1)^{-1}   
 -\frac{M!}{(m-1)!(M-m)!} \sum^{m-1}_{l=0}c_l(-1)^k\\\nonumber &\times 
\sum^{\infty}_{k=1}\frac{
 (M-m+l+1)^{k-1}  (\eta-1)^k}{k! \rho^k_m}
  .
\end{align}
By applying Lemma \ref{lemma2} and also using the fact that $\rho_m$ approaches infinity, $T_2$ can be approximated as follows:
 \begin{align}
T_2\approx &
 -\frac{M!}{(m-1)!(M-m)!} \sum^{\infty}_{k=1}(-1)^k\frac{
  (\eta-1)^k}{k! \rho_m^k} 
\sum^{k-1}_{q=0}{k-1 \choose q}(M-m+1)^{k-1-q}\sum^{m-1}_{l=0}c_l l^q
  .
\end{align}
Again applying \eqref{q2}, $T_2$ can be approximated as follows: 
 \begin{align}
T_2\approx  &
 \frac{ M!}{ (M-m)!} \frac{
  (\eta-1)^m}{m!\rho_m^m} .
\end{align}
One can observe that the decay rate of $T_1$ is $\rho_m^{-\frac{m}{2}}$, but the decay rate of $T_2$ is $\rho_m^{-m}$. Therefore, at high SNR, $T_1$ is dominant and  the proof for the lemma is complete.

\section{Proof for Lemma \ref{lemmalow}}
Depending on whether  $\rho_n> \rho_m$ holds, the low SNR approximation for $\mathrm{P}_n   $ can be obtained differently, as shown in the following subsections. 
\subsection{For the case of $\rho_n\leq \rho_m$}
In this case, $T_2=0$, the probability $\mathrm{P}_n  $ is given by
\begin{align}
\mathrm{P}_n   
=& 
c_{mn} \sum^{n-1-m}_{p=0} \frac{c_p}{M-m-p}\sum^{m-1}_{l=0} c_l e^{\frac{b^2}{4a}}    \frac{\sqrt{\pi}}{2\sqrt{a}}\left(1 - \Phi\left(\frac{b\sqrt{a}}{2a}\right)\right).
\end{align}

Recall the probability integral function can be approximated as follows:
\begin{align}
\Phi(x) \approx 1 - \frac{e^{-z^2}}{\sqrt{\pi}z}\sum^{k_x}_{k=0} (-1)^k\frac{(2k-1)!!}{(2z^2)^k} ,
\end{align}
for $x\rightarrow \infty$, where $k_x$ decides how many terms to be kept for the approximation. At low SNR, i.e.,  $\rho_m\rightarrow 0$, $\frac{b\sqrt{a}}{2a}$ also approaches infinity, and therefore,  the probability $\mathrm{P}_n  $ can be approximated as follow: 
\begin{align}
\mathrm{P}_n   
\approx& 
c_{mn} \sum^{n-1-m}_{p=0} \frac{c_p}{M-m-p}\sum^{m-1}_{l=0} c_l      \frac{1}{ b}\sum^{k_x}_{k=0} (-1)^k\frac{(2k-1)!! 2^ka^k}{b^{2k}} \\\nonumber 
\underset{(7)}{\rightarrow}& 
c_{mn} \sum^{n-1-m}_{p=0} \frac{c_p}{M-m-p}\sum^{m-1}_{l=0}     \frac{c_l }{ M-m+l+1}  \underset{(8)}{=}1 ,
\end{align}
where step $(7)$ follows by using $k_x=0$, and  step $(8)$ follows from the following fact
\begin{align}
1
=& \int^{\infty}_0 \int^{\infty}_{x} f_{|h_m|^2, |h_n|^2}(x,y) dydx=
c_{mn} \sum^{n-1-m}_{p=0} \frac{c_p}{M-m-p}\sum^{m-1}_{l=0}     \frac{c_l }{ M-m+l+1}.
\end{align}

\subsection{For the case of $\rho_n> \rho_m$}
Recall that the probability $\mathrm{P}_n   $ is the sum of the two terms, $T_1$ and $T_2$. For the case of   $\rho_n\geq \rho_m$,  $T_1$ is given by 
\begin{align} 
T_1
=& 
c_{mn} \sum^{n-1-m}_{p=0} \frac{c_p}{M-m-p}\sum^{m-1}_{l=0} c_l e^{\frac{b^2}{4a}}   \frac{\sqrt{\pi}}{2\sqrt{a}}\left(1 - \Phi\left(\frac{\eta-1}{\rho_m}+\frac{b\sqrt{a}}{2a}\right)\right) ,
\end{align}
At low SNR, i.e.,  $\rho_m\rightarrow 0$, we have the following approximation: 
\begin{align} 
\frac{\eta-1}{\rho_m}+\frac{b\sqrt{a}}{2a} =& \frac{\eta-1}{\rho_m}+ \frac{p+l+1+\frac{M-m-p}{\eta}}{2\sqrt{\frac{\rho_m}{\eta} (M-m-p)}} \approx  \frac{\eta-1}{\rho_m} \rightarrow \infty. 
\end{align}
 Again applying the approximation of the probability integral function, in the low SNR regime, the probability can be approximated as follows:
\begin{align}\label{eqlemma1x}
T_1   \approx &  
c_{mn} \sum^{n-1-m}_{p=0} \frac{c_p}{M-m-p}\sum^{m-1}_{l=0} c_l    \frac{e^{\frac{b^2}{4\frac{\rho_m}{\eta} (M-m-p)}-\left(\frac{\eta-1}{\rho_m} \right)^2}}{2\sqrt{\frac{\eta-1}{\eta\rho_m} (M-m-p)}}   \rightarrow 0 ,
\end{align}
where we set $k_x=0$. The last approximation follows from the facts that $\rho_m^{-2}$ is more dominant than $\rho_m^{-1}$ for $\rho_m\rightarrow 0$, and $x^{\frac{1}{2}}e^{-x}\rightarrow 0$ for $x\rightarrow \infty$. It is easy to show that $T_2\rightarrow 1$ since
\begin{align}
 \frac{M!}{(m-1)!(M-m)!} \sum^{m-1}_{l=0}c_l\frac{e^{-(M-m+l+1)\frac{\eta-1}{\rho_m} }}{M-m+l+1}\rightarrow 0,
\end{align}
 for   $\rho_m\rightarrow 0$. 
Since $T_1\rightarrow 0$ and $T_2\rightarrow 1$, we have   $\mathrm{P}_n \rightarrow 1$.  

Therefore, no matter whether  $\rho_n> \rho_m$,  $\mathrm{P}_n$ always approaches $1$ and the proof for the lemma is complete.  

\section{Proof for Lemma \ref{lemmaqmn}}

With some algebraic manipulations,   the probability $\tilde{\mathrm{P}}^D_n$ can be written as follows:  
\begin{align}\label{xd}
\tilde{\mathrm{P}}^D_n= &\mathrm{P} \left( |g_m|^2\leq \frac{\epsilon}{\rho} \right)+ \mathrm{Q}_{mn} ,
\end{align}
where 
\begin{align}
\mathrm{Q}_{mn} \triangleq &\mathrm{P} \left( |g_m|^2> \frac{\epsilon}{\rho}, T\log\left(1+\alpha_n^2\rho|g_n|^2\right)+T\log(1+ \tilde{\beta}\rho |g_n|^2) \leq T\log(1+\rho |g_n|^2) \right). 
\end{align}
Note that in \eqref{xd}, we use the fact that when $|g_m|^2\leq \frac{\epsilon}{\rho}$, MEC server $n$ cannot be admitted during the first time slot and hence its rate in NOMA is always smaller than that of OMA due to $\beta<1$.

By using the marginal pdf of $|g_m|^2$, $\mathrm{P} \left( |g_m|^2\leq \frac{\epsilon}{\rho} \right) $ can be   calculated as follows:
\begin{align} \label{xd2}
 \mathrm{P} \left( |g_m|^2\leq \frac{\epsilon}{\rho} \right) = 1-\frac{K!\sum^{m-1}_{l=0}c_l\frac{e^{-(K-m+l+1)\frac{\epsilon}{\rho} }}{K-m+l+1}}{(m-1)!(K-m)!}  .
\end{align}

The second term in \eqref{xd}, denoted by $\mathrm{Q}_{mn}$, can be rewritten  as follows:
\begin{align}\nonumber
\mathrm{Q}_{mn}  =  &\mathrm{P} \left( |g_m|^2> \frac{\epsilon}{\rho}, \left(1+ \frac{\rho |g_m|^2 -\epsilon}{\rho |g_m|^2(1+\epsilon)}\rho|g_n|^2\right) (1+ \tilde{\beta}\rho |g_n|^2) \leq  (1+\rho |g_n|^2) \right),
\end{align}
where the equation follows by using the    CR power allocation coefficient in \eqref{cr}. 
With some algebraic manipulations, the term $\mathrm{Q}_{mn} $ can be expressed as follows: 
\begin{align}\nonumber
\mathrm{Q}_{mn} =  &\mathrm{P} \left( |g_m|^2> \frac{\epsilon}{\rho},   |g_n|^2\leq \frac{\rho |g_m|^2[(1-\tilde{\beta})(1+\epsilon)-1]+\epsilon}{\rho\tilde{\beta}(\rho|g_m|^2-\epsilon)} \right). 
\end{align}

Due to the channel ordering assumption made in \eqref{order2}, we have the following inequality 
\begin{align} 
  |g_m|^2\leq  |g_n|^2\leq \frac{\rho |g_m|^2[(1-\tilde{\beta})(1+\epsilon)-1]+\epsilon}{\rho\tilde{\beta}(\rho|g_m|^2-\epsilon)} ,
\end{align}
which leads to the following constraint on $ |g_m|^2$:
\begin{align} \label{cc}
   |g_m|^2\leq    \frac{\rho |g_m|^2[(1-\tilde{\beta})(1+\epsilon)-1]+\epsilon}{\rho\tilde{\beta}(\rho|g_m|^2-\epsilon)} .
\end{align}
With some algebraic manipulations, one can find that $-\frac{1}{\rho}$ and $\frac{\epsilon}{\tilde {\beta}\rho}$ are the two roots of the following quadratic form:
\begin{align}  
\rho\tilde{\beta}(\rho x-\epsilon)x-  \left( \rho x[(1-\tilde{\beta})(1+\epsilon)-1]+\epsilon\right) =0.
\end{align}
 
Therefore,  the constraint in \eqref{cc} can be surprisingly written in a very simplified form as follows:
\begin{align} \label{ccx}
    |g_m|^2 \leq \frac{\epsilon}{\tilde {\beta}\rho} .
\end{align}
Note that $\tilde{\beta}\leq 1$, which means $ \frac{\epsilon}{\tilde {\beta}\rho}\geq
 \frac{\epsilon}{ \rho}$. 
Therefore $\mathrm{Q}_{mn}$ can be further expressed as follows:
\begin{align}\nonumber
\mathrm{Q}_{mn} =  &\mathrm{P} \left(  \frac{\epsilon}{\rho}<|g_m|^2\leq  \frac{\epsilon}{\tilde{\beta}\rho},  |g_n|^2\leq \frac{\rho |g_m|^2[(1-\tilde{\beta})(1+\epsilon)-1]+\epsilon}{\rho\tilde{\beta}(\rho|g_m|^2-\epsilon)} \right) ,
\end{align}
where we use the fact that
\[
\mathrm{P} \left( |g_m|^2>  \frac{\epsilon}{\tilde {\beta}\rho}, |g_n|^2\leq \frac{\rho |g_m|^2[(1-\tilde{\beta})(1+\epsilon)-1]+\epsilon}{\rho\tilde{\beta}(\rho|g_m|^2-\epsilon)} \right)=0.
\]
After applying the joint pdf in \eqref{joint}, the term $\mathrm{Q}_{mn} $ can be written as follows: 
\begin{align}\nonumber
\mathrm{Q}_{mn} =  & 
c_{mn} \sum^{n-1-m}_{p=0}c_p\int^{ \frac{\epsilon}{\tilde{\beta}\rho}}_{\frac{\epsilon}{\rho}}e^{-(p+1)x} (1-e^{-x})^{m-1} \int^{\frac{\rho x[(1-\tilde{\beta})(1+\epsilon)-1]+\epsilon}{\rho\tilde{\beta}(\rho x-\epsilon)} }_{x}e^{-(K-m-p)y} dydx
\\\nonumber
  =& 
c_{mn} \sum^{n-1-m}_{p=0}c_p\int^{ \frac{\epsilon}{\tilde{\beta}\rho}}_{\frac{\epsilon}{\rho}}e^{-(p+1)x} (1-e^{-x})^{m-1}    \frac{e^{-(K-m-p)x}-e^{-(K-m-p)\frac{\rho x[(1-\tilde{\beta})(1+\epsilon)-1]+\epsilon}{\rho\tilde{\beta}(\rho x-\epsilon)} }}{K-m-p} dx.
\end{align}
After applying the Chebyshev-Gauss approximation, $\mathrm{Q}_{mn}$ can be approximated as follows: 
\begin{align} \label{qmn}
\mathrm{Q}_{mn} \approx  & c_{mn} \sum^{n-1-m}_{p=0}c_p\sum^{N}_{i=1} \frac{\pi}{N} \left(\frac{\epsilon}{2\tilde{\beta}\rho}- \frac{\epsilon}{2\rho}\right)\\\nonumber &\times f\left(\left(\frac{\epsilon}{2\tilde{\beta}\rho}+ \frac{\epsilon}{2\rho}\right) +\left(\frac{\epsilon}{2\tilde{\beta}\rho}- \frac{\epsilon}{2\rho}\right)\theta_i \right) \sqrt{1-\theta_i^2}. 
\end{align}
By substituting \eqref{xd2} and \eqref{qmn} into \eqref{xd}, the proof for the lemma is complete. 

\section{Proof for Lemma \ref{lemma5}}

 Recall that $\tilde{\mathrm{P}}^D_n$ is the sum of two terms, i.e., $
\tilde{\mathrm{P}}^D_n= \mathrm{P} \left( |g_m|^2\leq \frac{\epsilon}{\rho} \right)+ \mathrm{Q}_{mn}$. By using  the proof for Lemma \ref{lemmahigh}, the first part of $\tilde{\mathrm{P}}^D_n$ can be approximated as follows:
 \begin{align}
 \mathrm{P} \left( |g_m|^2\leq \frac{\epsilon}{\rho} \right) =&1-\frac{K!\sum^{m-1}_{l=0}c_l\frac{e^{-(K-m+l+1)\frac{\epsilon}{\rho} }}{K-m+l+1}}{(m-1)!(K-m)!} \\\nonumber \approx &
 \frac{ K!}{ (K-m)!} \frac{
  2^{2m}(1-\beta)^m}{m!\beta^{2m}\rho^m} \doteq \rho^{-m}. 
\end{align}

In the following, the approximation of $Q_{mn}$ will be focused. 
According to the mean value theorem for integrals, $\tilde{\mathrm{P}}^D_n$ can be evaluated as follows: 
\begin{align}\nonumber
\mathrm{Q}_{mn} =&  c_{mn} \sum^{n-1-m}_{p=0}c_pe^{-(p+1) \frac{\epsilon}{\xi\rho}} (1-e^{- \frac{\epsilon}{\xi\rho}})^{m-1}   \frac{e^{-\frac{\xi_1}{\rho}}-e^{-\frac{\xi_2}{\rho}}}{K-m-p}, 
\end{align}
for a parameter $\xi$ satisfying 
\begin{align}
  \frac{\epsilon}{\tilde{\beta}\rho} \leq \frac{\epsilon}{\xi\rho} \leq \frac{\epsilon}{\rho},
\end{align}
where $\tilde{\beta}\leq \xi \leq 1$.

To simplify the notation, we define $\xi_1=(K-m-p) \frac{\epsilon}{\xi }$ and $\xi_2=(K-m-p)\frac{   [(1-\tilde{\beta})(1+\epsilon)-1]+\xi}{ \tilde{\beta}(  1-\xi)} $. Note that both the parameters, $\xi_1$ and $\xi_2$, are not functions of the SNR. Therefore, $\tilde{\mathrm{P}}^D_n$ can be approximated as follows:
\begin{align}\nonumber
\mathrm{Q}_{mn} \approx&  c_{mn} \sum^{n-1-m}_{p=0}c_p  \sum^{m-1}_{l=0}c_l e^{-(p+1) \frac{\epsilon}{\xi\rho}} e^{- \frac{l\epsilon}{\xi\rho}}   \frac{e^{-\frac{\xi_1}{\rho}}-e^{-\frac{\xi_2}{\rho}}}{K-m-p}  \\\nonumber
 =&    c_{mn} \sum^{n-1-m}_{p=0}c_p e^{-(p+1) \frac{\epsilon}{\xi\rho}}    \frac{e^{-\frac{\xi_1}{\rho}}-e^{-\frac{\xi_2}{\rho}}}{K-m-p} \sum^{m-1}_{l=0}c_l e^{- \frac{l\epsilon}{\xi\rho}}.
\end{align}
 By applying the series representation for the exponential function, $\tilde{\mathrm{P}}^D_n$ can be approximated as follows:
\begin{align}\nonumber
\mathrm{Q}_{mn} \approx&      c_{mn} \sum^{n-1-m}_{p=0}c_p e^{-(p+1) \frac{\epsilon}{\xi\rho}}    \frac{e^{-\frac{\xi_1}{\rho}}-e^{-\frac{\xi_2}{\rho}}}{K-m-p} \sum^{m-1}_{l=0}c_l  \sum^{\infty}_{k=0} \frac{l^k\epsilon^k}{\xi^k\rho^k k!}
\\\nonumber = &   c_{mn} \sum^{n-1-m}_{p=0}c_p e^{-(p+1) \frac{\epsilon}{\xi\rho}}    \frac{e^{-\frac{\xi_1}{\rho}}-e^{-\frac{\xi_2}{\rho}}}{K-m-p}\sum^{\infty}_{k=0}   \frac{\epsilon^k\sum^{m-1}_{l=0}c_l l^k}{\xi^k\rho^k k!}.
\end{align}
Now applying the properties in \eqref{q1} and \eqref{q2}, we have the following approximation
\begin{align}\nonumber
\mathrm{Q}_{mn} \approx&    c_{mn} \sum^{n-1-m}_{p=0}c_p e^{-(p+1) \frac{\epsilon}{\xi\rho}}    \frac{e^{-\frac{\xi_1}{\rho}}-e^{-\frac{\xi_2}{\rho}}}{K-m-p}    \\\nonumber &\times \frac{\epsilon^{m-1}(-1)^{m-1}(m-1)! }{\xi^{m-1}\rho^{m-1} (m-1)!}.
\end{align}
In order to remove the sum with respect to $p$, we first rewrite $\tilde{\mathrm{P}}^D_n$ as follows:
\begin{align}\nonumber
\mathrm{Q}_{mn} \approx&    c_{mn} \frac{\epsilon^{m-1}(-1)^{m-1}(m-1)! }{\xi^{m-1}\rho^{m-1} (m-1)!}     \sum^{n-1-m}_{p=0}c_p e^{-(p+1) \frac{\epsilon}{\xi\rho}}    \frac{e^{-\frac{\xi_1}{\rho}}-e^{-\frac{\xi_2}{\rho}}}{K-m-p} \\ 
 \approx&   \frac{1}{\rho^{m}}  \frac{c_{mn}\epsilon^{m-1}(-1)^{m-1}(m-1)! }{\xi^{m-1} (m-1)!}    \sum^{n-1-m}_{p=0}c_p      \frac{ \xi_2 - \xi_1  }{K-m-p} . 
\end{align}
In order to make  Lemma \ref{lemma2} applicable, the sum in $\tilde{\mathrm{P}}^D_n$ can be first rewritten as follows: 
\begin{align}\nonumber
\mathrm{Q}_{mn} \approx&      \frac{1}{\rho^{m}}  \frac{c_{mn}\epsilon^{m-1}(-1)^{m-1}(m-1)! }{\xi^{m-1} (m-1)!}   \\\nonumber &\times(\xi_2 - \xi_1 )\sum^{n-1-m}_{t=0}{n-1-m\choose t}      \frac{  (-1)^t   }{K-n+1+t} 
\\\nonumber
\underset{(9)}{=}&    \frac{1}{\rho^{m}}  \frac{c_{mn}\epsilon^{m-1}(-1)^{m-1}(m-1)! }{\xi^{m-1} (m-1)!}   \\  &\times(\xi_2 - \xi_1 )\frac{(n-m-1)!(K-n)!}{(K-m)!} \doteq \rho^{-m},
\end{align}
where step (9) follows by using Lemma \ref{lemma2}. 
Since both the terms in \eqref{xccc} have the same order of $m$, the proof for the lemma is complete.
  \bibliographystyle{IEEEtran}
\bibliography{IEEEfull,trasfer}

  \end{document}